\documentclass[twocolumn, pra, nofootinbib, aps]{revtex4}

\usepackage{graphicx}
\usepackage{bm}
\usepackage{makeidx} 
\usepackage{amsmath}
\usepackage{amssymb}
\usepackage{underscore}
\usepackage{color}
\usepackage{amsthm}
\usepackage{hyperref}
\usepackage{cleveref}
\usepackage{diagbox}
\hypersetup{colorlinks, citecolor=magenta, filecolor=blue, linkcolor=blue, urlcolor=green}
\usepackage[normalem]{ulem}
\usepackage{float}
\usepackage{array}
\newcolumntype{P}[1]{>{\centering\arraybackslash}p{#1}}
\newtheorem{theorem}{Theorem}

\begin{document}

%\title{Intrinsic squashed non-locality as an upper bound on device independent key}

\title{Computable genuine multimode entanglement measure: Gaussian vs. non-Gaussian}

\author{Saptarshi Roy\(^1\), Tamoghna Das\(^{2}\), Aditi Sen(De)\(^1\)}

\affiliation{\(^1\)Harish-Chandra Research Institute, HBNI, Chhatnag Road, Jhunsi, Allahabad 211 019, India}
\affiliation{\(^2\)International Centre for Theory of Quantum Technologies,  University of Gda\'{n}sk, 80-952 Gda\'{n}sk, Poland.}

\begin{abstract}
%Genuine entanglement in continuous variable systems is quantified via the generalized geometric measure (GGM). For multimode Gaussian states, a closed form expression of GGM  is computed in terms of the symplectic invariants of the reduced states. Examples of typical three and four mode Gaussian states are presented. 
%In the non-Gaussian paradigm, we compute GGM for photon added or subtracted bright squeezed vaccum and four mode squeezed vacuum states and find that both addition and subtraction of photons lead to enhancement of the genuine entanglement content of the state. Our analysis reveal that photon subtraction almost always results in a larger enhancement of GGM compared to photon addition for both single mode and multimode operations. Furthermore, we observe a novel freezing feature of GGM during some specific cases of photon subtraction.
Genuine multimode entanglement in continuous variable systems can be quantified by exploring the geometry of the state-space, namely via the
generalized geometric measure (GGM). It is defined by  the shortest distance of a given multimode state from a nongenunely multimode entangled state. For the multimode Gaussian states, we derive a closed form expression of GGM in terms of the symplectic
eigenvalues of the reduced states. Following that prescription, the characteristics  of GGM for  typical three-  and four-mode Gaussian states
are investigated. In the non-Gaussian paradigm, we compute GGM for photon-added as well as -subtracted states having three- and four-modes and find that both addition and subtraction of photons enhance the genuine multimode entanglement  of the state compared to its Gaussian counterpart. 
Our analysis reveals that when an initial three-mode vacuum state is evolved according to an interacting Hamiltonian, photon addition is more 
beneficial in increasing GGM compared to photon subtraction while 
the scenerio reverses when one considers the four mode non-Gaussian states.
Specifically, subtracting photons from  four-mode squeezed vacuum states almost always result in higher multimode entanglement content than that of photon
addition to both single as well as  multimode  and constrained as well as
unconstrained operations. Furthermore, we observe 
that GGM freezes under subtraction of photons involving multiple modes, in some specific cases. This feature is novel in its own rights as it does not appear while adding photons.
Finally, we  relate the  enhancements of GGM with the distance-based non-Gaussianity measure.
\end{abstract}

\maketitle

\section{Introduction}
\label{sec:intro}
Among the plethora of nonclassical features, intrinsic to quantum mechanics, entanglement \cite{hhhh}, a term coined by Schrodinger  \cite{schrodinger}, is arguably the most fascinating one. Entanglement, captures the degree of inseparability 
of quantum systems,  turns out to be the resource exhibiting 
 ``quantum advantage and supremacy" in various quantum information processing tasks 
 %\textcolor{red}{encompassing quantum communication protocols} 
 like teleportation \cite{tele1,tele2,tele3}, dense coding \cite{dc,dcrev}, entanglement-based quantum cryptography \cite{crypto1,crypto2} and the detection of quantum phase transitions \cite{qpt1,qpt2,qptbook1,qptbook2}, to name a few.

Classification of quantum systems based on entanglement is one of the premiere endeavors in quantum information science. 
%\textcolor{red}{this line I changed completely}
Although the problem of quantifying entanglement of a given system (state) turns out to be more intricate with the increase of the number of parties. However, if one restricts the analysis only for pure states, the categorization  becomes somehow simpler. For example, bipartite pure states can either be entangled or product. For multiple parties, even within the set of pure states, we can have diverse possibilities, where states can be entangled for some of the parties while product with the rest. A prototypical instance of such a situation can be illustrated with an example of a three-qubit state, $|\eta\rangle_{ABC} = |\psi^-\rangle_{AB}\otimes|\phi\rangle_C$, commonly known as a biseparable state, where $|\psi^-\rangle_{AB}$ is the maximally entangled state and $|\phi\rangle_{C}$ is any arbitrary pure qubit. 
Note that $|\eta\rangle_{ABC}$ is product in $AB:C$, and the reduced $A:C$ as well as $B:C$ bipartitions, while it is entangled in  other bipartite cuts. The pure quantum states that are entangled in all bipartitions  are called genuinely multiparty entangled.
% Naturally, they come in both pure and mixed variants, but in this paper, we would consider only pure genuinely entangled states. 
Typically, different kinds of entanglement present in multiparty systems can be broadly quantified in two ways: using distance-based (geometric) measures \cite{geo1, geo11, geo2, geo3, geo5, geo4, geo6, geo7, geo8, geo9, geo10}, and by using monogamy-based measures \cite{gm1,gm2, gm3, gm4, gm5, gm6, gm7, gm8, gm9, gm10,gm11,gm12}. 
%In this work, we would consider a distance-based measure to quantify the genuine multiparty content of infinite dimensional systems \cite{contvar-rev, contvar-book, gaussian3}.
In finite dimensions, genuine multipartite entanglement (GME) for pure states can be computed efficiently by a distance-based measure called the generalized geometric measure (GGM) \cite{ggm1} (see also \cite{ggm22,ggm3,ggm4,ggm5}), and has been used extensively to study genuine multipartite entanglement for a wide range of quantum systems \cite{gs1,gs2,gs4,gs3,gdc1}. %\textcolor{red}{(prefer to use "for a wide range of quantum systems")}.  
Although attempts have been made to generalize GGM for mixed states \cite{ggm2}, the analysis is not exhaustive and only works for some 
specific classes of states. 

In this paper, we are interested in investigating the entanglement between various modes of multimode states of light, where each mode contains an arbitrary number of photons. Hence, the corresponding Hilbert space is an infinite dimensional one \cite{contvar-rev, contvar-book, gaussian3}.
%content present in the multimode states, where each mode is an infinite dimensional system \cite{contvar-rev, contvar-book, gaussian3}. 
We first show that, like in finite dimensions, the GGM can  still be simplified in terms of the Schmidt coefficients. However, 
%On the other hand, if one moves on to infinite dimensional systems , although the formula for GGM holds good for systems with finite energies, which we prove in this paper, 
this route is not  very efficient for continuous variable systems 
%to compute the GGM for the obvious reason of 
since the operations involve infinite dimensional matrices.
For Gaussian states \citep{gaussian1,gaussian2}, we show that such problems can be resolved 
%\textcolor{red}{(bypass I think weakened it)} 
by expressing it in a closed analytical form in terms of the symplectic invariants, namely the symplectic eigenvalues of the covariance matrix. These invariants can be easily computed from the  quadratic quadrature correlations, which form the elements of the covariance matrix. Since the dimension of the covariance matrix grows linearly with the number of modes, our method provides an efficient and scalable prescription for computing GGM of pure Gaussian states having an arbitrary number of modes. 
 
% \textcolor{red}{comment about the computational simplicity; put some comments on scalability} 
We demonstrate the applicability of our recipe by computing the GGM for some prototypical Gaussian states of three- and four-modes. 
Examples include three- and four-mode squeezed states in which  GGM increases with the increase of the squeezing strength. 
We also show an instance of generating  genuine multimode entanglement from  an initially uncorrelated three-mode vacuum state which is evolved by an interacting Hamiltonian describing a nonlinear crystal. 
We want to stress here that our prescription is lucrative from an experimental point of view since the symplectic eigenvalues can be computed from the data of the quadrature correlations composing the covariance matrix.

We then move on to study genuine multimode entanglement in non-Gaussian states. 
Although the Gaussian states offer several advantages like elegant mathematical simplicity in the description and easy experimental realizability.
It has also been established that non-Gaussian resources are ``more" useful compared to their Gaussian counterparts for some of the  quantum information protocols including  
%In case of non-Gaussian states, such simplification cannot be carried out and we only have the canonical formula of GGM at our disposal. Nevertheless, it can still be used to compute the GGM of some non-Gaussian states which
 %are well known resource in the field of quantum information processing. 
bosonic codes \cite{bc1}, pioneering application in photonic quantum computation \cite{qcomp},
%Its fundamental role in several other branches including 
quantum metrology \cite{qmetro}, entanglement distillation \cite{Ent-distillation, Giedke-ent-distillation}, entanglement distribution \cite{entdistribution},  error correction \cite{gaussian-error}, phase estimation \cite{phase-estimation},  quantum communication \cite{qcommunication} and quantum cloning \cite{qcloning}.

In general, any state whose Wigner function \cite{Wigner} is not  Gaussian in the quantum phase space is considered as a non-Gaussian state. One of the popular methods
 to obtain non-Gaussian states is to add or subtract photons in different modes of a  Gaussian state. It has also been demonstrated that photon addition and subtraction make a negative dip \cite{prabhu-usha} in the Gaussian Wigner function of a Gaussian state, while  the entanglement of the photon-added and -subtracted states are always monotonically increasing with the number of photon added and subtracted \cite{TMSV, FMSV}. 

In this paper, we also study whether de-Gaussification leads to the enhancement of GGM from its Gaussian value. Among the different de-Gaussification techniques, we choose the method of photon addition and subtraction  since it can be scalably applied when the the number of modes of the initial Gaussian state grows and  can also be achieved experimentally \cite{exp1, exp2}.
 Specifically, we add or subtract photons from different modes of the three- and four-mode Gaussian states and track the enhancement of GGM.
%Although the canonical formula for the GGM is inefficient, it can still be used to compute the GGM for some specific classes of non-Gaussian states where the symmetry of the state gives some information about the location of the maximal eigenvalue, thereby drastically reducing the search space. 
%%%A premiere example of such states include the photon-added or subtracted four mode squeezed vacuum state (FMSV) \cite{FMSV}.
 In our context of computation of genuine multimode entanglement, we find that both adddition and subtraction of photons in one or several modes of the squeezed vacuum state lead to enhancement in the GGM content compared to its Gaussian counterpart. Further analysis reveal that photon-subtraction leads to a higher multimode entanglement  than that of the photon-added state.
  Moreover, we choose the relative entropy-based non-Gaussianity measure to quantify the departure of a given state from the initial Gaussian states  and show that it increases monotonically with the increase of number of photons added (subtracted). We, however, notice that such a non-Gaussianity measure fails to shed light on the increment  obtained for genuine multimode entanglement content of  photon-subtracted state over the photon-added ones. 
% It was shown in previous works by some of us \cite{FMSV} that subtracted state can sometimes possess higher entanglement in some specific bipartition .... (why are we including this? isnt this expected?)
 
This work is organized as follows. After giving the definition of GGM in continuous variable systems in Sec. \ref{sec:prelim}, we show that the computation of GGM can be simplified both for Gaussian and non-Gaussian  states.
%a brief introduction in Sec. \ref{sec:intro}, we arrive at Sec. \ref{sec:prelim} where  we use Schmidt decomposition in continuous variable systems to evaluate the canonical formula for the GGM in the infinite dimensional case and present a proof of the same. 
We then evaluate the expression of GGM for Gaussian states in terms of its symplectic invariants in Sec. \ref{sec:gaussian} and obtain GGM for some typical examples in Sec. \ref{sec:examples-gaussian}. The results involving non-Gaussian states are given in Sec. \ref{sec:nongaussian} before presenting the conclusion in Sec. \ref{sec:con}. 

%atleast in its canonical form, 

\section{Preliminaries}
\label{sec:prelim}
We introduce the notion of genuine multimode entanglement measure based on the geometry of multimode quantum states. The generalized geometric measure (GGM) for an arbitrary $N$-mode pure  quantum state, $|\psi_{12\ldots N}\rangle$,  is defined as 
\begin{equation}
\label{Eq:GGM}
{\cal G}(|\psi_{12\ldots N}\rangle) = 1 - \max_{|\chi\rangle \in nG} |\langle \chi |\psi_{12\ldots N}\rangle|^2,
\end{equation}
where the maximization is taken over the set of all $N$-mode pure states $|\chi\rangle$, which are not genuinely multimode entangled, denoted by \(nG\).
The distance measure used in this case is  the  Fubini study metric \cite{fubini1,fubini2}.
Based on Schmidt decomposition for continuous variable systems \cite{schmidt,schmidt2} which can be easily extended for any normalizable infinite dimensional state (Apendix \ref{app:app1} for details), we obtain an expression of GGM in terms of the eigenvalues of the reduced density matrices. We call it to be the ``canonical" form of GGM.  

We then show in the next section that for Gaussian states, a much more elegant and efficient form of GGM in terms of symplectic eigenvalues of the reduced modes can be obtained. On the other hand, the canonical form of GGM can be important to quantify multimode entanglement for non-Gaussian states where the simplifications as obtained in the Gaussian case is unavailable.
 and we have to resort to the brute force method.

 %For multimode Gaussian states, it has been rigorously proved in \cite{schmidt}. The idea can be easily extended for any general normalizable infinite dimensional state as well, see Apendix for details. We use it to tailor the ``canonical" expression 
%of GGM in terms of the eigenvalues of the reduced density matrices, just as in the case of finite dimensional systems. Naturally one at once realizes that the canonical formula is really inefficient since it would involve maximization over infinite number of eigenvalues. We would 
%show later that for Gaussian states, we can express the GGM much more elegantly and efficiently in terms of the symplectic eigenvalues of the reduced modes. But first, we would present a proof 
%for the canonical expression of the GGM, which would actually turn out to be useful for certain classes of non-Gaussian states where the simplifications as obtained in the Gaussian case is unavailable and we have to resort to the brute force method.

\subsection{The canonical formula of GGM}
\label{sec:ggm-proof}
%Now we try to prove the simplified form of GGM, i.e., the maximal eigenvalues for the CV system.
%Suppose we have an $N$ mode pure quantum state $|\psi_{12\ldots N}\rangle$
We now use Schmidt decomposition to simplify the evaluation of GGM given in Eq. \eqref{Eq:GGM}. Since $|\chi\rangle$ is nongenuinely  multimode entangled, i.e., it is product with respect to at least one modal-bipartition,  we can always write the state in the Schmidt decomposition across that partition as  
\begin{equation}
|\chi\rangle = |\chi_{12\ldots N}\rangle = |\chi_{\cal A}\rangle \otimes |\chi_{\cal B}\rangle,
\end{equation} 
where we assume that ${\cal A}$ and ${\cal B}$ contains $n(\mathcal{A})$ and $n(\mathcal{B})$ modes respectively, such that $n(\mathcal{A})+n(\mathcal{B})=N$.
%where ${\cal A}\cup {\cal B} = 1,2,\ldots ,N$ and  ${\cal A} \cap {\cal B} = \emptyset$. 
By using the Schmidt decomposition  in the same ${\cal A}:{\cal B}$ modal-bipartition, we can write the given state $|\psi_{12\ldots N}\rangle$ of $n(\mathcal{A})+n(\mathcal{B})$ modes as
%Here we have assume that $|\chi_{12\ldots N}\rangle$ is product in the ${\cal A}:{\cal B}$ bipartition.
%We can write the given state $|\psi_{12\ldots N}\rangle$, by using the Schmidt decomposition  in the same ${\cal A}:{\cal B}$ bipartition, as 
\begin{equation}\label{Eq:schmidt}
|\psi_{12\ldots N}\rangle = \sum_i \sqrt{\lambda_i} |\mu_{\cal A}^i\rangle \otimes  |\nu_{\cal B}^i\rangle,
\end{equation}
where $\{\lambda_i\}$ are the set of Schmidt coefficients which are always positive and $\sum_i \lambda_i = 1$. Here $\{|\mu_{\cal A}^i\rangle\}$ and $\{|\nu_{\cal B}^i\rangle\}$ are the Schmidt basis (also form a basis \footnote{If  the Schmidt basis does not span the entire space i.e., it  forms a basis in their respective Hilbert space, one can  include more mutually orthonormal states, so that it spans the entire space, thereby form a basis.}
 in the Hilbert spaces of $\cal A$ and $\cal B$) and $i$  runs upto  $\min\{\text{dim} {\cal H^A}, \text{dim} {\cal H^B}\}$. One can also expand $|\chi_{\cal A}\rangle$ and $|\chi_{\cal B}\rangle$ in terms of the corresponding Schmidt basis, given by
 \begin{equation}
 |\chi_{\cal A}\rangle = \sum_i a_i |\mu_{\cal A}^i\rangle,
 \end{equation}
 and 
 \begin{equation}\label{Eq:chi-Bexpand}
 |\chi_{\cal B}\rangle = \sum_i b_i |\nu_{\cal B}^i\rangle,
\end{equation}
 where $\sum_i  |a_i|^2 = 1$ and $\sum_i  |b_i|^2 = 1$. 
Using Eqs. (\ref{Eq:schmidt}) --  (\ref{Eq:chi-Bexpand}), we can rewrite the second term in the right hand side of Eq.  (\ref{Eq:GGM}) as 
% Putting these forms into the second term of Eq. (\ref{Eq:GGM}), we get
\begin{widetext}
 \begin{eqnarray}
 \max_{|\chi\rangle \in nG} |\langle \chi |\psi_{12\ldots N}\rangle| &=& \max_{\{a_i\}, \{ b_j\}} |\sum_i   \sum_j  a_i^* b_j^*  \langle \mu_{\cal A}^i| \otimes \langle \nu_{\cal B}^j| \sum_k \sqrt{\lambda_k} |\mu_{\cal A}^k\rangle \otimes  |\nu_{\cal B}^k\rangle| \nonumber \\
 &=& \max_{\{a_i\}, \{ b_j\}} |\sum_k   a_k^* b_k^* \sqrt{\lambda_k} | \label{Eq:GGM-proof} \\ 
 &\leq & \max_{\{a_i\}, \{ b_j\}} \sum_k   |a_k| |b_k| \sqrt{\lambda_k}~.  \label{Eq:triangular} 
 \end{eqnarray}
 \end{widetext}
To obtain 
the inequality (\ref{Eq:triangular}),  we use the triangle inequality\footnote{$|a+b| \leq |a| + |b|$.}.
 The optimization over all nongenuinely multimode entangled states $|\chi\rangle$ now reduces to the optimization over the state parameters, $\{a_i\}$ and $\{b_j\}$.
If we assume that the Schmidt coefficients  $\{\lambda_i\}$s are arranged in the descending order, 
 we have
\begin{eqnarray}
\max_{|\chi\rangle \in nG} |\langle \chi |\psi_{12\ldots N}\rangle| &\leq & \sqrt{\lambda_1} \max_{\{a_i\}, \{ b_j\}} \sum_k |a_k| |b_k|  \nonumber \\
&\leq & \sqrt{\lambda_1} \max_{\{a_i\}, \{ b_j\}} \sqrt{\sum_i |a_i|^2 \sum_j |b_j|^2} \nonumber \\
&\leq & \sqrt{\lambda_1} \label{Eq:final-GGM}.
\end{eqnarray}
 The second inequality is obtained by using the well known Cauchy-Schwarz inequality, and the inequality (\ref{Eq:final-GGM}) is due to the normalization conditions in terms of $a_i$ and $b_j$. 
By choosing $|a_1| = |b_1| = 1$ and the rest of the coefficients to be $0$ in Eq. (\ref{Eq:GGM-proof}),  the above bound can be achieved, and hence the GGM of $|\psi_{12\ldots N}\rangle$ reduces to 
\begin{eqnarray}
\label{Eq:GGM-final-express}
&&\mathcal{G}(|\psi_{12\ldots N}\rangle) = \nonumber \\  &&1 - \max \big \lbrace \lambda_{\cal A:B} | {\cal A}\cup {\cal B} = \lbrace 1,2,\ldots, N \rbrace, {\cal A}\cap{\cal B} = \emptyset \big \rbrace, \nonumber \\
\end{eqnarray}
where $\lambda_{\cal A:B}$ is the maximal Schmidt coefficient in the $\cal A : B$
modal split of $|\psi_{12\ldots N}\rangle$, and maximization is performed over all such possible mode-bipartitions. Equipping ourselves with the canonical formula of GGM, we now proceed to compute the GGM for pure multimode Gaussian states in terms of the symplectic invariants.
%
%
%\vspace{1in}
%
%\textcolor{red}{In this proof we use triangular inequality and Cauchy Schwarz inequality for countably infinite indices.
%We know that the following relation holds
%\begin{equation}
%\left |\int_0^1 f(x)dx \right| \leq \int_0^1 |f(x)| dx,
%\end{equation}
%and 
%\begin{eqnarray}
%\int_0^1 f(x)g(x)dx \leq \sqrt{\int_0^1 f^2(x) dx \int_0^1 g^2(x) dx},
%\end{eqnarray}
%So what we use is true or not.}

\section{GGM for pure multimode Gaussian states}
\label{sec:gaussian}
The covariance matrix of an arbitrary $m$-mode Gaussian state, $\rho$,   is a $2m \times 2m$ matrix, $\Lambda$, defined by 
\begin{eqnarray}
\Lambda_{ij} = \frac{1}{2} \big \langle \lbrace R_i,R_j \rbrace \big \rangle - \langle R_i  \rangle \langle R_j \rangle,
\end{eqnarray}
where $\vec{R} = (q_1,p_1,q_2,p_2, ... q_m,p_m)^{\text{T}}$ with $q_i$s and $p_i$s being the usual quadrature operators. The quadrature operators are related to the field operators, $a_i$s, in the following way:
\begin{eqnarray}
q_j = \frac{1}{\sqrt{2}}(a_j + a_j^{\dagger}), ~p_j = \frac{1}{\sqrt{2}i}(a_j - a_j^{\dagger}),
\end{eqnarray}
where $i = \sqrt{-1}$. 
 %Any valid covariance matrix satisfies the  bonafide condition, given by
 The positivity of $\rho$ can be certified from the bonafied condition on the covariance matrix
\begin{eqnarray}
\Lambda+iJ \geq 0, \text{ where } J = \bigoplus_{i=1}^{m} 
\begin{bmatrix}
 0 & 1 \\
-1 & 0 
\end{bmatrix}.
\end{eqnarray}
Here $J$ is the symplectic matrix. Following Williamson's theorem \cite{contvar-rev,ferraro}, we notice that the covariance matrix $\Lambda$ can be obtained from $\Lambda^d$ by appropriate symplectic transformation ($\textbf{S}_\Lambda$),
\begin{eqnarray}
\Lambda = \textbf{S}_\Lambda \Lambda^d \textbf{S}_\Lambda^{\text{T}}
\end{eqnarray}
 with 
\begin{eqnarray}
\Lambda^d = \begin{bmatrix}
 \nu_1 \mathbb{I}_2 &  &  &  & & &  \\
 & \nu_2 \mathbb{I}_2 &   &  & & &  \\
&  &   & . &  & &  \\
&  &   &  & & . &  \\
  & & & & & & \nu_m \mathbb{I}_2
\end{bmatrix} = \bigoplus_{i=1}^m \nu_i \mathbb{I}_2, 
\end{eqnarray}
where $\{\nu_i\}$s are the symplectic eigenvalues of $\Lambda$ and $\mathbb{I}_2$ is the $2 \times 2$ identity matrix. Note that $\Lambda^d$ corresponds to a product of thermal states
\begin{eqnarray}
\rho^d = \bigotimes_{i=1}^m \sigma_{\beta_i}, \text{ with temperatures } \beta_i = \ln \frac{\nu_i + 1/2}{\nu_i - 1/2}.
\end{eqnarray}  
\begin{figure}[t]
\includegraphics[width=0.8\linewidth]{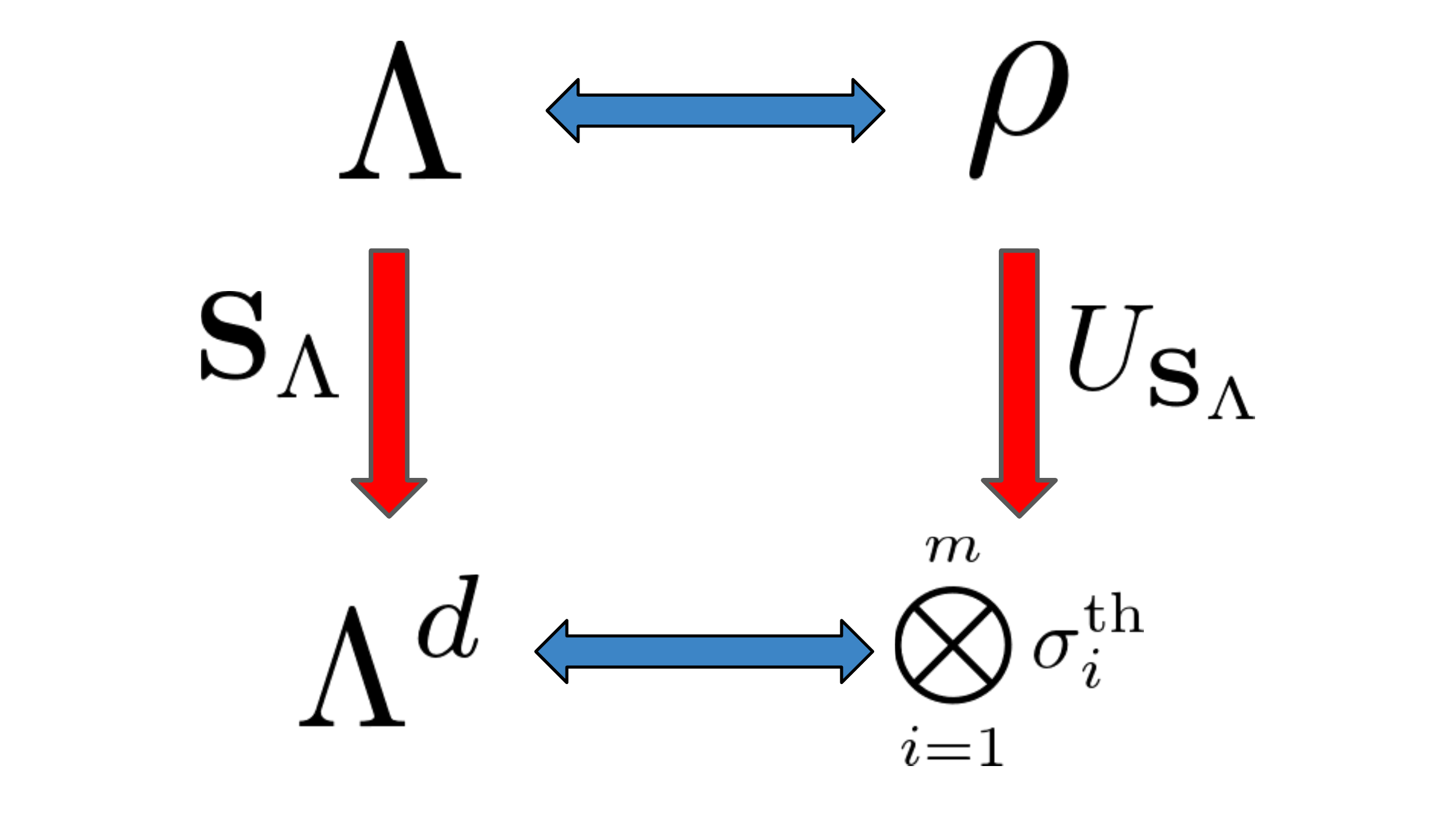}
\caption{Schematic representation of symplectic diagonalization. Any arbitrary Gaussian state can be written as a product of thermal states obtained by applying the appropriate unitary operation.}
\label{fig:schematic}
\end{figure}
The state $\rho$ is related to $\rho^d$ via the relation $\rho = U_{\textbf{S}_\Lambda}\rho^d U_{\textbf{S}_\Lambda}^\dagger$, where $U_{\textbf{S}_\Lambda}$ is the unitary operator corresponding to the symplectic transform $\textbf{S}_\Lambda$. Note that since $U_{\textbf{S}_\Lambda}$ is unitary, the eigenspectrum of $\rho$ and $\rho^d$ are identical. The above analysis is schematically depicted in Fig. \ref{fig:schematic}.
In general, the eigenvalues, $\lambda_n$, of a thermal state with inverse temperature $\beta$ is given by $\lambda_n = e^{-\beta n}(1 -e^{-\beta})$. The maximal eigenvalue simply corresponds to $\lambda_{n=0}$. Now, the eigenvalues of $\rho^d$, which are also eigenvalues of $\rho$ are then given by \begin{eqnarray}
\lambda_{n_1,n_2, . . .n_m}=\prod_i \lambda_{n_i}.
\label{eq:eig-m-mode}
\end{eqnarray}
The maximal eigenvalue of $\rho$ is therefore given by 
\begin{eqnarray}
\lambda_{n_1=0, n_2=0,...n_m = 0} = \prod_{i=1}^m (1- e^{-\beta_i}) = \prod_{i=1}^m \frac{2}{1 + 2 \nu_i}, 
\label{eq:eig-symplectic}
\end{eqnarray}
where $\nu_i$ is the symplectic eigenvalue of the $i^{\text{th}}$ mode. Expressing the maximal eigenvalue in terms of the symplectic specra enables evaluation of GGM in a much more efficient way which we encapsulate below in the form of the following theorem:
\begin{theorem}\label{th:GGM-new-form}
The GGM ($\mathcal{G}$) of a N-mode pure Gaussian state $|\psi_{12... N}\rangle$ is given by
\begin{eqnarray}\label{eq:GGM-new-form}
\mathcal{G}(|\psi_{12...N}\rangle) = 1 - \max \mathcal{P}_m \Big\lbrace \prod_{i=1}^m \frac{2}{1 +2 \nu_i} \Big\rbrace_{m=1}^{\big[\frac{N}{2}\big]},
\end{eqnarray}
where $\mathcal{P}_m$ denotes all the reduced states of $|\psi\rangle_{12 ...N}$ with $m$-modes, and $[x]$ denotes the integral part of $x$.
\end{theorem}
\begin{proof}
The canonical formula for the GGM in Eq. \eqref{Eq:GGM-final-express} reveals that the computation of GGM is equivalent to the evaluation of the maximum eigenvalue of all the reduced states of the given state $|\psi_{12...N}\rangle$. The canonical formula can be restated as 
\begin{eqnarray}
\mathcal{G}(|\psi_{12...N}\rangle) = 1 - \max \mathcal{P}_m \{\lambda_{n_1, n_2, ... n_m}  \}_{m=1}^{\big[ \frac{N}{2}\big]}, 
\end{eqnarray}
where $\lambda_{n_1, n_2, ... n_m}$s denote the eigenvalues of a particular $m$-mode reduced state of $|\psi_{12...N}\rangle$, and $\mathcal{P}_m$ denotes all the reduced states of $|\psi_{12...N}\rangle$ with $m$-modes.

%In Eq. \eqref{eq:eig-symplectic} we have established the relationship between the maximum eigenvalue to the symplectic spectrum of any $m$-mode reduction of $|\psi\rangle_{12...N}$ which 
From Eq. \eqref{eq:eig-symplectic}, we know that for a particular $m$-mode reduced state with spectra $\nu_i$s, the maximal eigenvalue simply reads $\prod_{i=1}^m \frac{2}{1 + 2\nu_i}$. The formula of GGM is then obtained by maximizing the maximal eigenvalue of all the $m \in \big[1, [N/2]\big]$-mode reductions of $|\psi\rangle_{12 ...N}$, and therefore we have
\begin{eqnarray}
\mathcal{G}(|\psi_{12...N}\rangle) = 1 - \max \mathcal{P}_m \Big\lbrace \prod_{i=1}^m \frac{2}{1 + 2\nu_i} \Big\rbrace_{m=1}^{\big[\frac{N}{2}\big]},
\end{eqnarray}
and hence the proof.
\end{proof}
We, therefore,  arrive at our goal of expressing the GGM of a pure multimode Gaussian state in terms of the symplectic spectrum of its reduced states. Below, we provide a prescription for computing the GGM for an arbitrary multimode pure state $|\psi_{12...N}\rangle$ in terms of the following steps which illustrates the simplicity of the method:
\begin{enumerate}
\item Evaluate all the $m$-mode reduced density matrices, $\rho_m$ of $|\psi_{12...N}\rangle$, where $m \in \big[1,[\frac{N}{2}]\big]$. Note that there are $\binom{N}{m}$ different $m$-mode reductions of the initial $N$-mode state.

\item Perform symplectic diagonalization of all the $\rho_m$s to get the diagonal $\rho_m^d$s and the corresponding symplectic eigenvalues, $\nu_i^m$s,  of each $\rho_m$.

\item Compute the maximal eigenvalue, $\bar{\lambda}_m^{\max}$, of each $\rho_m$ using the formula 
\begin{equation}
\bar{\lambda}_m^{\max} = \prod_{i=1}^{m} \frac{2}{1+2\nu_i}.
\end{equation}
Choose the maximum eigenvalue $\lambda_m^{\max}$ out of all the $\bar{\lambda}_m^{\max}$s.

\item The GGM of $|\psi\rangle_{1, 2, ...N}$ is then calculated from the following expression
\begin{equation}
\mathcal{G}(|\psi_{12...N}\rangle) = 1 - \max \lbrace \lambda_m^{\max} \rbrace_{m=1}^{[\frac{N}{2}]}.
\end{equation}
\end{enumerate}
We want to highlight here that the evaluation of GGM for multimode Gaussian states is much simpler compared to that for a multiqubit state. This is so because the $m$-mode reductions in case of Gaussian states are characterized by $2m \times 2m$  covariance matrices while the $m$-party reduced states for a multiqubit state is a $2^m \times 2^m$ dimensional density matrix. This difference in computational complexity (polynomial vs. exponential) implies that via exact diagonalization, we can atmost compute GGM of a $\sim 24$-qubit state, whereas we can go, at least in principle, upto $2^{12}$-modes.
Ultimately this number would be restricted by the number of reduced density matrices, $\binom{N}{m}$, for which we have to compute the eigenvalues for every $m$-mode reductions of the given $N$-mode pure state. This fact remains true in the multiqubit case also.
 Furthermore, partial tracing at the covariance matrix level is much easier compared to partial tracing at the level of states. 
Hence, our analysis provides an efficiently computable and scalable method to calculate the genuine multimode entanglment content of the  Gaussian states.   %illustrate the working of our formula by taking some typical examples.
%\textcolor{red}{We can make some interesting comments about computational complexity. Calculating GGM for multimode Gaussian states is computationally simpler than the evaluation of GGM of multiqubit states}. In the next section we apply our method to compute the GGM of some typical three and four mode Gaussian states.  

As an aside, we want to recall that GGM is a measure of geometric origin which, as mentioned earlier, is defined as the minimum distance of the given state (whose genuine multiparty entanglement content has to be calculated) from the set of non-genuinely multiparty entangled states. Hence, by definition, it is applicable for all states, both pure and mixed. However, for pure states, owing to Schmidt decomposition, the minimization problem becomes tractable and we have the canonical formula for GGM. In the case of mixed states, such simplification is unavailable. As shown for finite dimensional systems, if the states have some symmetries \cite{ggm2, gmixed}, the optimization can ``sometimes" be performed. Note that even for pure states, computation of GGM becomes challenging with the increase in the number of parties and their dimensions, let alone when the dimension of each party becomes infinite, the case considered in our work. This is due to the exponential growth of the size and number of partitions one has to diagonalize to compute GGM.

Nevertheless,  in the next section, we  use the above described prescription to compute the patterns of GGM for some typical Gaussian states and illustrate their multimode entanglement numerically.

\subsection{GGM of some typical Gaussian states}
\label{sec:examples-gaussian}
We start our analysis with a three-mode pure Gaussian state. The first example is a three-mode state prepared by combining three single-mode squeezed states in a “tritter” (a three-mode generalization of a beam-splitter) \cite{ferraro}. Secondly, we track the dynamics of GGM of states generated in a single nonlinear crystal. Next we consider the case of how well can bimodal entanglement be  distributed among three modes by means of passive operations using beam splitters. In the genre of four modes, we calculate the GGM of the four-mode squeezed vaccum with respect to the squeezing parameter.

\begin{figure}[h]
\includegraphics[width=0.9\linewidth]{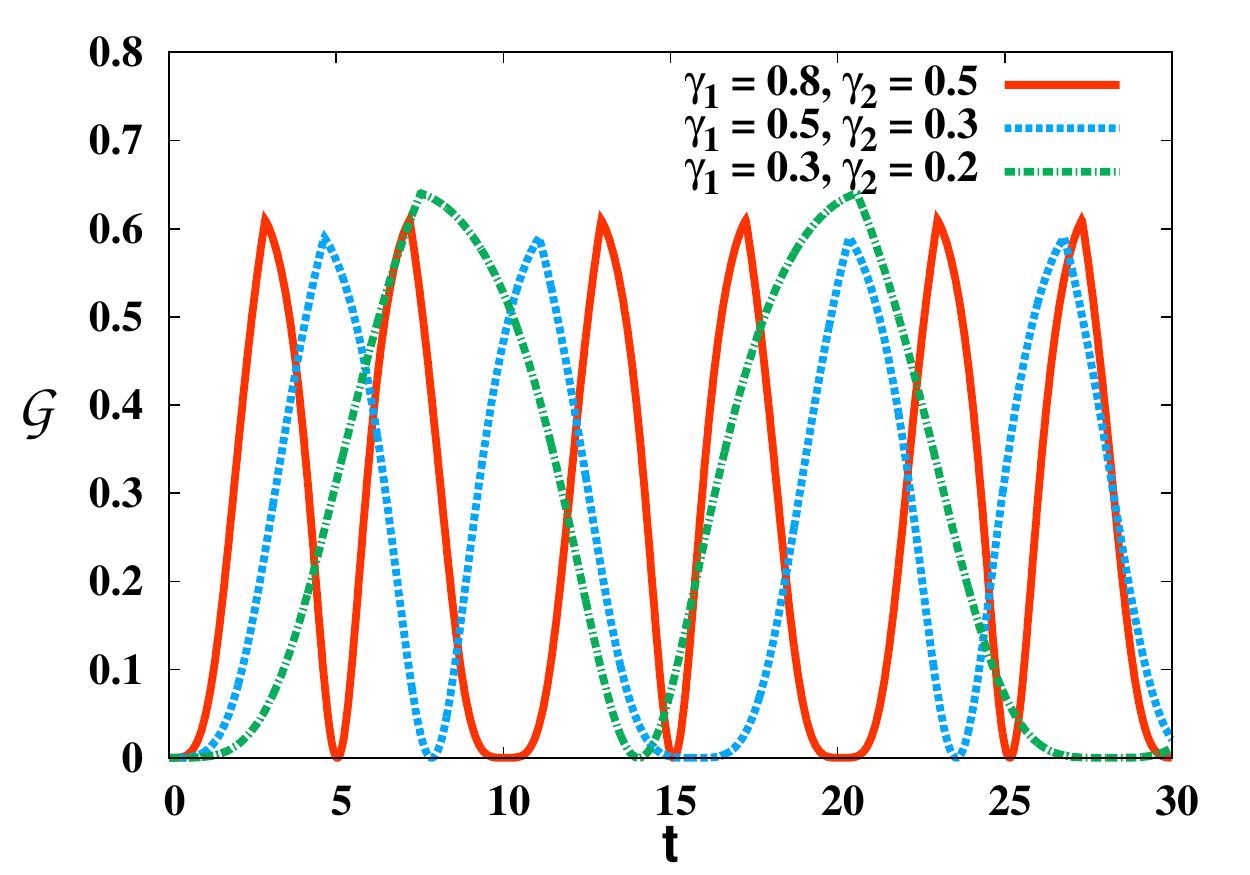}
\caption{(Color online) GGM of  $|\psi^b_3 (t) \rangle$ vs. time. The state is generated according to the Hamiltonian $H_I$, given in Eq. \eqref{eq:hamt}. Solid, dotted and dashed lines are for different choices of $\gamma_1$ and $\gamma_2$ values.}
\label{fig:g1}
\end{figure}

\subsubsection{Three-mode Gaussian states}
\emph{Example 1.} Three single-mode squeezed vacuum states of strength $r$ combined using the ``tritter" (a three-mode generalization of the beam splitter) gives a three-mode state, $|\psi_3^a\rangle$, which  possess the covariance matrix, given by \cite{ferraro}
%The first state which we consider is $|\psi_3^a\rangle$, and is constructed using the ``tritter" by combining three single mode squeezed vacuum of strength $r$. It possess the following covariance matrix
%\begin{eqnarray}
%\Lambda^a = \frac{1}{2}
%\begin{bmatrix}
%\mathcal{R}_+ & \mathcal{S} & \mathcal{S} & 0 & 0 & 0 \\
%\mathcal{S} & \mathcal{R}_+ & \mathcal{S} & 0 & 0 & 0 \\
%\mathcal{S} & \mathcal{S} & \mathcal{R}_+ & 0 & 0 & 0 \\
%0 & 0 & 0 & \mathcal{R}_- & -\mathcal{S} & -\mathcal{S} \\
%0 & 0 & 0 &  -\mathcal{S} & \mathcal{R}_- & -\mathcal{S} \\
%0 & 0 & 0 &  -\mathcal{S} &  -\mathcal{S} & \mathcal{R}_- 
%\end{bmatrix},
%\end{eqnarray}

\begin{eqnarray}
\Lambda^a = \frac{1}{2}
\begin{bmatrix}
\mathcal{R}_+ & 0 & \mathcal{S} & 0 & \mathcal{S} & 0 \\
0 & \mathcal{R}_- & 0 & -\mathcal{S} & 0 & -\mathcal{S} \\
\mathcal{S} & 0 & \mathcal{R}_+ & 0 & \mathcal{S} & 0 \\
0 & -\mathcal{S} & 0 & \mathcal{R}_- & 0 & -\mathcal{S} \\
\mathcal{S} & 0 & \mathcal{S} &  0 & \mathcal{R}_+ & 0 \\
0 & -\mathcal{S} & 0 &  -\mathcal{S} &  0 & \mathcal{R}_- 
\end{bmatrix}
\end{eqnarray}
where $\mathcal{R}_\pm = \cosh 2r \pm \frac{1}{3} \sinh 2r$ and $\mathcal{S} = -\frac{2}{3}\sinh 2r$.
It turns out to be a symmetric state with identical reduced single-mode covariance matrices ($\Lambda_1^a = \Lambda_2^a = \Lambda_3^a$), which reads as
\begin{eqnarray}
\Lambda_1^a = \frac{1}{2}
\begin{bmatrix}
\mathcal{R}_+ & 0 \\
0 & \mathcal{R}_-
\end{bmatrix}.
\end{eqnarray}
Following the prescription for calculating the GGM, we note that for three-mode states, we   only need to consider the single-mode reductions ($\big[ \frac{3}{2} \big] = 1$). The symplectic eigenvalues of $\Lambda_1^a$ is given by  
\begin{eqnarray}
\nu_1 = \frac{1}{2}\sqrt{\mathcal{R}_+\mathcal{R}_-} = \frac{1}{6}\sqrt{5+4\cosh 4r},
\end{eqnarray}
%The GGM then reads as
leading to GGM as
\begin{eqnarray}
\mathcal{G}(|\psi^a_3\rangle) = 1- \frac{2}{1+2\nu_1} = \frac{\frac{1}{3}\sqrt{5+4\cosh 4r}-1}{\frac{1}{3}\sqrt{5+4\cosh 4r}+1}.
\end{eqnarray}
We, therefore, obtain an expression of  GGM as a function of the squeezing strength $r$. Expectedly, as $r \rightarrow \infty$, the GGM approaches its algebraic maximum of unity.

\emph{Example 2.} Let us consider a three-mode vacuum state induced by an interaction Hamiltonian, $H_I$, describing action  of a single nonlinear crystal 
\begin{eqnarray}
H_I = \gamma_1 a_1^\dagger a_3^\dagger + \gamma_2 a_2^\dagger a_3 + h.c.,
\label{eq:hamt}
\end{eqnarray}
where $a_i^\dagger$ and $a_i$ represent the creation and the anhilation operators respectively. The effective coupling constants $\gamma_k,k= 1,2$, are functions of the nonlinear susceptibilities and the pump intensities.
The Hamiltonian $H_I$ describes interactions among three modes of the radiation field which are coupled via two parametric pumps. Such a process can also be realized exprimentally in $\chi^{(2)}$ media \cite{ferraro2}. 
We consider the three-mode vacuum state as the initial state and then the system evolves according to the Hamiltonian $H_I$, i.e. the time evolved three-mode state at any given time, $t$, can be represented as $|\psi^b_3 (t) \rangle = e^{-i H_I t}|0\rangle\otimes|0\rangle\otimes|0\rangle$. The closed form expression of the same can also be computed \cite{ferraro,ferraro2} as
% We compute the GGM dynamics of a three-mode vacuum state evolved by $H_I$. Such a state at any time $t$ reads as $|\psi^b_3 (t) \rangle = e^{-i H_I t}|0\rangle\otimes|0\rangle\otimes|0\rangle$. 
\begin{widetext}
\begin{eqnarray}\label{eq:3-mode}
|\psi^b_3 (t) \rangle = \frac{1}{\sqrt{1+n_1}} \sum_{r,s=0}^\infty \big( \frac{n_2}{1+n_1} \big)^{r/2} \big( \frac{n_3}{1+n_1} \big)^{s/2} e^{-i(r\phi_2 + s\phi_3)} \sqrt{\frac{(r+s)!}{r!s!}}|r+s,r,s\rangle,
\end{eqnarray}
\end{widetext}
where $n_i = \langle a_i^\dagger a_i \rangle$ denotes the average number of photons in the $i$-th mode with $\phi_j$s being the  phase factors.  For this Hamiltonian, we have $n_1 = n_2 + n_3$ at all times with 
\begin{eqnarray}
n_2 &=& \frac{|\gamma_1|^2|\gamma_2|^2}{\Omega^4}(1 - \cos \Omega t)^2, \nonumber \\ 
n_3 &=& \frac{|\gamma_1|^2}{\Omega^2}\sin^2 \Omega t,
\end{eqnarray}
%\begin{figure*}[ht]
%\includegraphics[width=\linewidth]{ggm2.pdf}
%\caption{GGM wrt to the transmission coefficients `s' and `t' of the two beam splitters.}
%\label{fig:ggm}
%\end{figure*}
where $\Omega = \sqrt{|\gamma_2|^2-|\gamma_1|^2}$. The covariance matrix, $\Lambda^b$, of the above state \cite{ferraro} is
%\begin{eqnarray}
%\Lambda^b = 
%\begin{bmatrix}
%\mathcal{F}_1 & \mathcal{A}_2 & \mathcal{A}_3 & 0 & -\mathcal{B}_2 & -\mathcal{B}_3 \\
%\mathcal{A}_2 & \mathcal{F}_2 & \mathcal{C} & -\mathcal{B}_2 & 0 & \mathcal{D} \\
%\mathcal{A}_3 & \mathcal{C} & \mathcal{F}_3 & -\mathcal{B}_3 & -\mathcal{D} & 0 \\
%0 & -\mathcal{B}_2 & -\mathcal{B}_2 & \mathcal{F}_1 & -\mathcal{A}_2 & -\mathcal{A}_3 \\
%-\mathcal{B}_2 & 0 & -\mathcal{D} &  -\mathcal{A}_2 & \mathcal{F}_2 & \mathcal{C} \\
%-\mathcal{B}_3 & \mathcal{D} & 0 &  -\mathcal{A}_3 &  \mathcal{C} & \mathcal{F}_3 
%\end{bmatrix},
%\end{eqnarray}
\begin{eqnarray}
\Lambda^b = 
\begin{bmatrix}
\mathcal{F}_1 & 0 & \mathcal{A}_2  & -\mathcal{B}_2 & \mathcal{A}_3 & -\mathcal{B}_3 \\
0 & \mathcal{F}_1 & -\mathcal{B}_2 & -\mathcal{A}_2 & -\mathcal{B}_3 & -\mathcal{A}_3 \\
\mathcal{A}_2 & -\mathcal{B}_2 & \mathcal{F}_2 & 0 & \mathcal{C} & \mathcal{D} \\
-\mathcal{B}_2 & -\mathcal{A}_2 & 0 & \mathcal{F}_2 & -\mathcal{D} & \mathcal{C} \\
\mathcal{A}_3 & -\mathcal{B}_3 & \mathcal{C} &  -\mathcal{D} & \mathcal{F}_3 & 0 \\
-\mathcal{B}_3 & -\mathcal{A}_3 & \mathcal{D} &  \mathcal{C} &  0 & \mathcal{F}_3 
\end{bmatrix},
\end{eqnarray}
where $\mathcal{F}_i = n_i + 1/2$, $\mathcal{A}_i = \sqrt{n_i(1+n_1)}\cos \phi_i$, $\mathcal{B}_i = \sqrt{n_i(1+n_1)}\sin \phi_i$, $\mathcal{C} = \sqrt{n_2 n_3}\cos (\phi_2 - \phi_3)$, and $\mathcal{D} = \sqrt{n_2 n_3}\sin (\phi_2 - \phi_3)$. The single-mode reduced covariance matrices $\Lambda^b_i$s in this case are
\begin{eqnarray}
\Lambda^b_i = \mathcal{F}_i 
\begin{bmatrix}
1 & 0 \\
0 & 1
\end{bmatrix}.
\end{eqnarray}
Notice that unlike in the previous case with $|\psi^a_3\rangle$, the single-mode reduced covariance matrices, $\Lambda^b_i$s, are not identical. Nevertheless, all of them come in  diagonalized form with symplectic eigenvalues, $\lbrace \mathcal{F}_1,\mathcal{F}_2,\mathcal{F}_3 \rbrace$. Hence the GGM reads as
\begin{equation}
\mathcal{G} (|\psi^b_3 (t) \rangle) = 1 - \max \Big \lbrace \frac{2}{1 + 2\mathcal{F}_i} \Big \rbrace_{i=1}^3,
\end{equation}
where $\mathcal{F}_i$s are functions of time.
For fixed $\gamma_i$s, the oscillatory behavior of GGM for the evolved state against time is depicted in Fig. \ref{fig:g1}. Note that the output state is independent of phase factors.
%Using the above formula, we can tract the evolution of GGM as a function of time (see Fig.  for some specific examples). 
The presence of kinks (nonanalyticities) in the dynamics is due to the maximization present in the formula of GGM. The times at which these kinks appear are precisely those where there is a crossover among the first two highest eigenvalues,  specifically, when the second 
 highest eigenvalue becomes the maximal one.
%one overtakes the the first one, thereby becomming the new maximual eigenvalue.

%Now, instead of combining three single mode squeezed vaccum as done for obtaining $|\psi^a_3\rangle$, here we consider another method of obtaining three mode pure states (for details of the process see [gerrardo]). In this case we start off with a two mode squeezed vacuum (TMSV) state and a single mode vacuum state and combine them with the help of three beam splitters. Note that in the preparation procedure so prescribed \cite{gerrardo}, the transmission coefficient of the beam splitter acting on modes $2$ and $3$ is fixed to $2/3$.  The transmissivities of the other two beam splitters are tunable.
%Now we ask the following question: How does the GGM of the post processed state depend on the transmission coeeficients of the remaining two beam splitters?

%\begin{figure}
%\includegraphics[width=0.9\linewidth]{fmsvbsv.pdf}
%\caption{GGM content of the FMSV and the BSV state with respect to the squeezing strength $r$.}
%\label{fig:fmsvbsv}
%\end{figure}

\subsubsection{Four-mode Gaussian states}
Let us now move to a computation of genuine multimode entanglement of a four-mode Gaussian state, namely the four-mode squeezed vacuum (FMSV) state \cite{FMSV}.  
%The bright squeezed emerges from the parametric down conversion of type II, and it is given by 
%\begin{eqnarray}
%|BSV\rangle = e^{r(a_H^\dagger b_V^\dagger - a_V^\dagger b_H^\dagger) + h.c.}|0000\rangle = \frac{1}{\cosh^2 r} \times \nonumber \\
%\sum_{n = 0}^\infty \tanh^n r   \sum_{s = 0}^n (-1)^s |n - s\rangle_{a_H} |s\rangle_{a_V} |s\rangle_{b_H} |n - s\rangle_{b_V} 
%\label{eq:bsv}
%\end{eqnarray}
%Where $a$ and $b$ are two optical fields and $H$, $V$ represents the polarization degrees of freedom of both the fields. In our approach we will consider the above state as an four modes state, i.e., the creation operators $a_H^\dagger$ and $a_V^\dagger$ of the same field, are two different modes. Here $r$ is the squeezing parameter which squeezed the quadrature operators. From the above expression it is prominent that the BSV state is highly symmetrical with the interchange of mode $1~(a_H)$ with $4~(b_V)$ as well as $2~(a_V)$ with $3~(b_H)$.
%
%It turns out that for the BSV state, the maximal eigenvalue always comes from the single mode reduced sectors. Furthermore, all the singlemode reduced covariance matrices, $\Lambda_i^{\text{BSV}}$s are identical and its symplectic eigenvalue reads as $\cosh^2 r - \frac{1}{2}$.
%%\begin{eqnarray}
%%\Lambda_i^{BSV} = 
%%\begin{bmatrix}
%%\cosh^2 r - \frac{1}{2} & 0 \\
%%0 & \cosh^2 r - \frac{1}{2}
%%\end{bmatrix}.
%%\end{eqnarray}
%The GGM of $|BSV\rangle$ then computed to be
%\begin{eqnarray}
%\mathcal{G}(|BSV\rangle) = 1 - \frac{2}{1+ 2(\cosh^2 r - 1/2)} = \tanh^2 r.
%\end{eqnarray}
The FMSV state can be obtained by using two single-mode squeezed vacuum passing through a $50$:$50$ beam splitter followed by another two beam splitter on the output modes of the previous one. The FMSV state can be represented as
\begin{eqnarray}\label{eq:FMSV}
|FMSV\rangle &=& e^{\frac r2\sum_{i=1}^4a_i^\dagger a_{i+1}^\dagger - a_i a_{i+1}}|0000\rangle  \nonumber \\
%&=& \frac{1}{\cosh r}e^{\frac 12 \tanh r(a_1^\dagger + a_3^\dagger)(a_2^\dagger + a_4^\dagger)}|0000\rangle 
&=& \frac{1}{\cosh r} \sum_{n = 0}^\infty \sum_{r_1 = 0}^n \sum_{r_2 = 0}^n \sqrt{{n}\choose{r_1}} \sqrt{{n}\choose{r_2}}  \nonumber\\
&& (\frac 12 \tanh r)^n |n - r_1\rangle |n - r_2\rangle | r_1\rangle | r_2\rangle,
\end{eqnarray}
where for $i=4$, $i+1$ is considered to be $1$.
%where $a^\dagger_i$ and $a_i$ are the bosonic creation and annihilation operators of the mode $i$ respectively. 
It is  clear from the above expression that the FMSV state is translationally invariant with the interchange of mode $1 \leftrightarrow 3$ and $2 \leftrightarrow 4$. Hence, we call modes $(1,2)$ and $(1,4)$ as adjacent modes while  $(1,3)$ and $(2,4)$ as alternate modes.

While computing GGM, the above symmetries greatly simplify the evaluation by reducing the search space during maximization.
In particular, we have to compute eigenvalues of the reduced covariance matrices  of a single-mode, $\Lambda_{\text{single}}^{\text{FMSV}}$, and two-mode reduced states consisting of adjacent and alternate modes, denoted respectively by $\Lambda_{\text{adjacent}}^{\text{FMSV}}$ and $\Lambda_{\text{alternate}}^{\text{FMSV}}$.
 %to only one single mode, and two dual mode reductions, namely the adjacent and the alternate modes. 
 For brevity, we do not give the expressions of the covariance matrix here (see Appendix \ref{app:app2} for their forms), $\Lambda^{\text{FMSV}}$, and its relevant reduced covariance matrices, $\Lambda_{\text{single}}^{\text{FMSV}}$, $\Lambda_{\text{adjacent}}^{\text{FMSV}}$, and $\Lambda_{\text{alternate}}^{\text{FMSV}}$. 
 %However, note that all of these can be computed quite easily. 
Nevertheless, the symplectic eigenvalue of  $\Lambda_{\text{single}}^{FMSV}$ is computed to be $\frac{1}{2}\cosh^2 r$, while those of $\Lambda_{\text{adjacent}}^{FMSV}$ and $\Lambda_{\text{alternate}}^{FMSV}$ are $\lbrace\frac{1}{2}\cosh r,\frac{1}{2}\cosh r \rbrace$ and $\lbrace \frac{1}{2}, \frac{1}{2}\cosh 2r \rbrace$ respectively. Therefore, the maximal eigenvalues from each of these sectors read as $\Big\lbrace \frac{2}{1+\cosh^2 r}, \frac{2}{1+\cosh 2r}, \Big(\frac{2}{1+\cosh r}\Big)^2 \Big\rbrace$, and hence the GGM of FMSV state can be expressed as 
\begin{widetext}
\begin{eqnarray}
\mathcal{G}(|FMSV\rangle) = 1 - \max \Big\lbrace \frac{2}{1+\cosh^2 r}, \frac{2}{1+\cosh 2r}, \Big(\frac{2}{1+\cosh r}\Big)^2 \Big\rbrace. 
\end{eqnarray}
\end{widetext}
%We now make a comparative study between the genuine multimode entanglement content of the BSV and the FMSV states. We find that the BSV possess more GGM compared to the FMSV for a given value of the squeezing parameter $r$, see Fig. \ref{fig:fmsvbsv}.
%\textcolor{red}{Here we will give the plot of the BSV and FMSV state with respect to the squeezing parameter.}
As expected, when the squeezing strength $r$ takes infinitely large values, the GGM of FMSV approaches unity.
%, see Fig \textcolor{blue}{put figure here}.

\section{GGM of non-gaussian states}
\label{sec:nongaussian}
As mentioned before, non-Gaussian resources can outperform their classical counterparts in a variety of quantum information and computation protocols. Therefore, we now venture into the non-Gaussian regime  and examine whether it offers any enhancement of GGM. Moreover, as already discussed, we choose to de-Gaussify states by photon addition and subtraction for their easy and scalable experimental implementability.
Note that such incremental behaviour on addition (subtraction) of photons was reported in the case of bipartite non-classical correlations for two and four-mode states \cite{TMSV, FMSV, enhance1, bell, enhance3, enhance2}. These results motivate us to consider the possibility of enhancement of genuine multimode entanglement on adding or subtracting photons from Gaussian states.

In the preceding section, we  calculated  GGM for several classes of three-mode and four-mode Gaussian states by using the form in terms of the symplectic eigenvalues of the covariance matrices derived in Eq. (\ref{eq:GGM-new-form}) of Theorem \ref{th:GGM-new-form}. 
%The calculation of the largest eigenvalues of all possible bipartition, and taking the supremum over all such eigenvalues boils down to a single supremum over all the Williamson normal form of each possible modewise bipartitions. 
However, such simplification does not work,  when one adds  (subtracts) photons in (from) Gaussian states. 
This is due to the fact that the photon-added (-subtracted) states do not represent Gaussian states and hence cannot be completely described by covariance matrices. 
% as the covariance matrix can not completely describe the newly formed state and one has to resort to brute force techniques to compute GGM.
%simpler form of Definition \ref{}, given in Eq. (\ref{}), where the Williamson normal form of the covariance matrix of the Gaussian states is needed. 
In this section, we use the cannonical form of GGM given in Eq. (\ref{Eq:GGM-final-express}) for the non-Gaussian states.
Our aim is to characterize GGM for several classes of non-Gaussian states, emerging from the three-mode state (given in Eq. (\ref{eq:3-mode})), and the FMSV state, by adding and subtracting photons.

It is worth mentioning that by the term \emph{photon addition} and \emph{subtraction} in different mode(s), we mean a consecutive addition in or subtraction of photons from that particular mode(s). It is indeed true that non-Gaussian states can also be obtained by several mixed ordering of photon addition and subtraction. But we will not be considering such cases here.

%For example one can consider first adding $m_i^1$ number of photons, followed by  $m_i^2$ number of photons subtraction then once again an addition of $m_i^3$ number of photons in the $i$th mode. Resulting a total either addition or subtraction of photons. In this paper we carefully restrains ourselves from such mixedness, and whenever we term addition we only means consecutive application of creation operator on that state, the same is true for subtraction also.

\begin{figure}[t]
\includegraphics[width=1\linewidth]{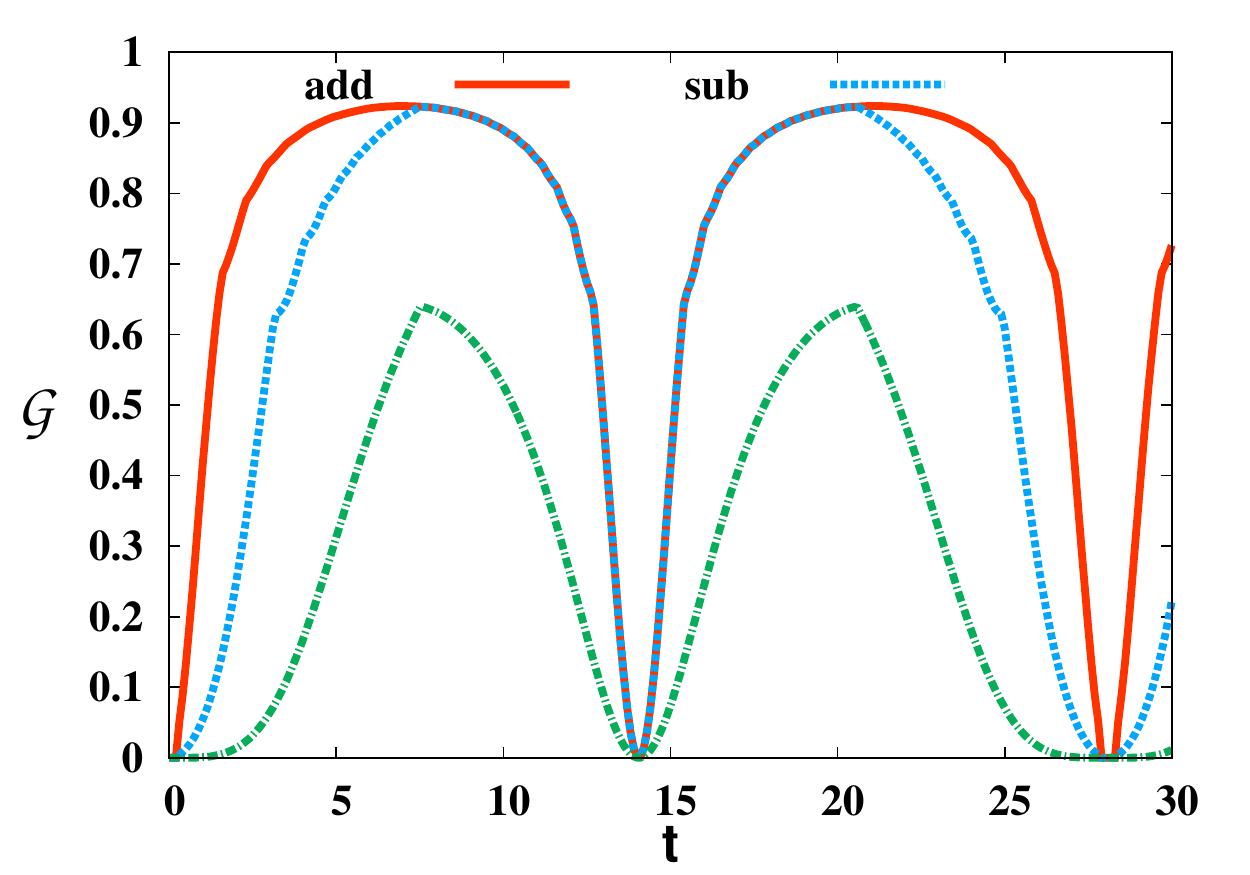}
\caption{GGM of the Gaussian as well as non-Gaussian states against time. In particular, we consider a three-mode Gaussian state $|\psi^b_3 (t) \rangle$ (green dot-dashed) and non-Gaussian states, $|\psi^b_3 (t)^{add\{m_i\}} \rangle$ (red solid) and $|\psi^b_3 (t) ^{sub\{m_i\}} \rangle$ (blue dashed), given in Eqs. (\ref{eq:3-mode-add}) and (\ref{eq:3-mode-sub}). Here $\gamma_1 = 0.8$ and $\gamma_2 = 0.5$. In the construction of non-Gaussian states $|\psi^b_3 (t)^{add\{m_i\}} \rangle $ and $|\psi^b_3 (t)^{sub\{m_i\}}  \rangle$, we choose $m_1 = 5, m_2 = 0$ and $m_3 = 0$. 
%he blue line points corresponds to ${\cal G}(|\psi^b_3 (t) \rangle)$, whereas the green 
%line point and purple line points are for ${\cal  G}(|\psi^b_3 (t) \rangle^{add\{m_i\}})$ and ${\cal  G}(|\psi^b_3 (t) \rangle^{sub\{m_i\}})$ respectively.
}
\label{fig:3-mode-nonG}
\end{figure}

Another important aspect of this photon addition and subtraction operation is that it is a physical operation on the given quantum state, based on the post-selection, i.e., it is a non-unitary operation. Here our main objective is to study genuine multimode entanglement content in these non-Gaussian states, where addition and subtraction of photons has been used as a tool to generate them. We then compare the GGM of photon-added (-subtracted) state with the GGM of the initial Gaussian state and address the question whether addition is better than subtraction from the perspective of amount of genuine multimode entanglement in these states. 
%But it demands a clarification that this derivation is indeed not an evolution of the Gaussian state.}--needs clarification

\subsection{Three-mode non-Gaussian state}

Let us first consider the non-Gaussian states derived from the three-mode Gaussian states, given in Eq. (\ref{eq:3-mode}), which are de-Gaussified via addition or subtraction of photons. The state takes the following form after adding  $m_i$ $(i= 1,2,3)$ number of photons in   mode $i$, denoted by \( |\psi^b_3 (t)^{add\{m_i\}} \rangle \) and  similarly the photon-subtracted  state, \(|\psi^b_3 (t)^{sub\{m_i\}} \rangle\): 
\begin{widetext}
\begin{eqnarray}
|\psi^b_3 (t)^{add\{m_i\}} \rangle &=& \frac{1}{\sqrt{N^{add}}} \sum_{r,s=0}^\infty \big( \frac{n_2}{1+n_1} \big)^{r/2} \big( \frac{n_3}{1+n_1} \big)^{s/2} e^{-i(r\phi_2 + s\phi_3)} \sqrt{\frac{(r+s)!}{r!s!}} \sqrt{\frac{(r+s + m_1)!}{(r + s)!}} \nonumber\\
&& \hspace{1.5in}\sqrt{\frac{(r+m_2)!}{r!}} \sqrt{\frac{(s + m_3)!}{s!}} ~|r+s + m_1\rangle r + m_2\rangle s + m_3\rangle, \label{eq:3-mode-add}\\
|\psi^b_3 (t)^{sub\{m_i\}} \rangle &=& \frac{1}{\sqrt{N^{sub}}} \sum_{\substack{r = m_2, s=m_3 \\ m_1 \geq m_2 + m_3}}^\infty \big( \frac{n_2}{1+n_1} \big)^{r/2} \big( \frac{n_3}{1+n_1} \big)^{s/2} e^{-i(r\phi_2 + s\phi_3)} \sqrt{\frac{(r+s)!}{r!s!}} \sqrt{\frac{(r+s)!}{(r + s - m_1)!}} \nonumber\\
&& \hspace{1.5in}\sqrt{\frac{r!}{(r - m_2)!}} \sqrt{\frac{s!}{(s - m_3)!}} ~|r+s - m_1\rangle r - m_2\rangle s - m_3\rangle, \label{eq:3-mode-sub}
\end{eqnarray}
\end{widetext}
where $N^{add}$ and $N^{sub}$ are the normalization constants of the photon-added and -subtracted states. ${\cal G}(|\psi^b_3 (t) \rangle^{add\{m_i\}} )$ (the purple line) and ${\cal G}(|\psi^b_3 (t) \rangle^{sub\{m_i\}} )$ (the green line) are plotted in Fig. \ref{fig:3-mode-nonG}, where we choose $m_1 = 5$ and $m_2 = m_3 = 0$, i.e.,  
$5$ photons are added  (subtracted) in (from) the first mode and no photons are added and subtracted from the rest of the modes. We observe that although the oscillatory behavior of GGM remains same, there is an enhancement of genuine multimode entanglement for both the non-Gaussian states (given in Eqs (\ref{eq:3-mode-add}) and (\ref{eq:3-mode-sub})), compared to their Gaussian counterpart, for all values of $t$. Moreover, we find that each oscillation contains two humps -- the first one is the mirror image of the second one. Specifically, we find that, in the first hump, GGM increases with $t$, and  ${\cal G}(|\psi^b_3 (t) \rangle^{add\{m_i\}} ) \geq {\cal G}(|\psi^b_3 (t) \rangle^{sub\{m_i\}} )$, thereby showing advantage of adding photons over subtraction. After ${\cal G}(|\psi^b_3 (t) \rangle^{add\{m_i\}} )$ and ${\cal G}(|\psi^b_3 (t) \rangle^{sub\{m_i\}} )$ reach their maxima, they coincide and start decreasing together with time. The behavior in the second hump images that of the first one. This example shows that by proper tuning of interaction strengths, photon-addition can enhance genuine multimode  entanglement upto $\sim 30 \%$ (as seen in Fig. \ref{fig:3-mode-nonG}). All the curves, obtained both from Gaussian and non-Gaussian states coincide at those times when $\mathcal{G} = 0$. This is because  in the first hump, bare Gaussian state itself becomes product at that time and we cannot generate multimode entangled states by adding or subtracting photons. 
\begin{figure}[t]
\includegraphics[width=\linewidth]{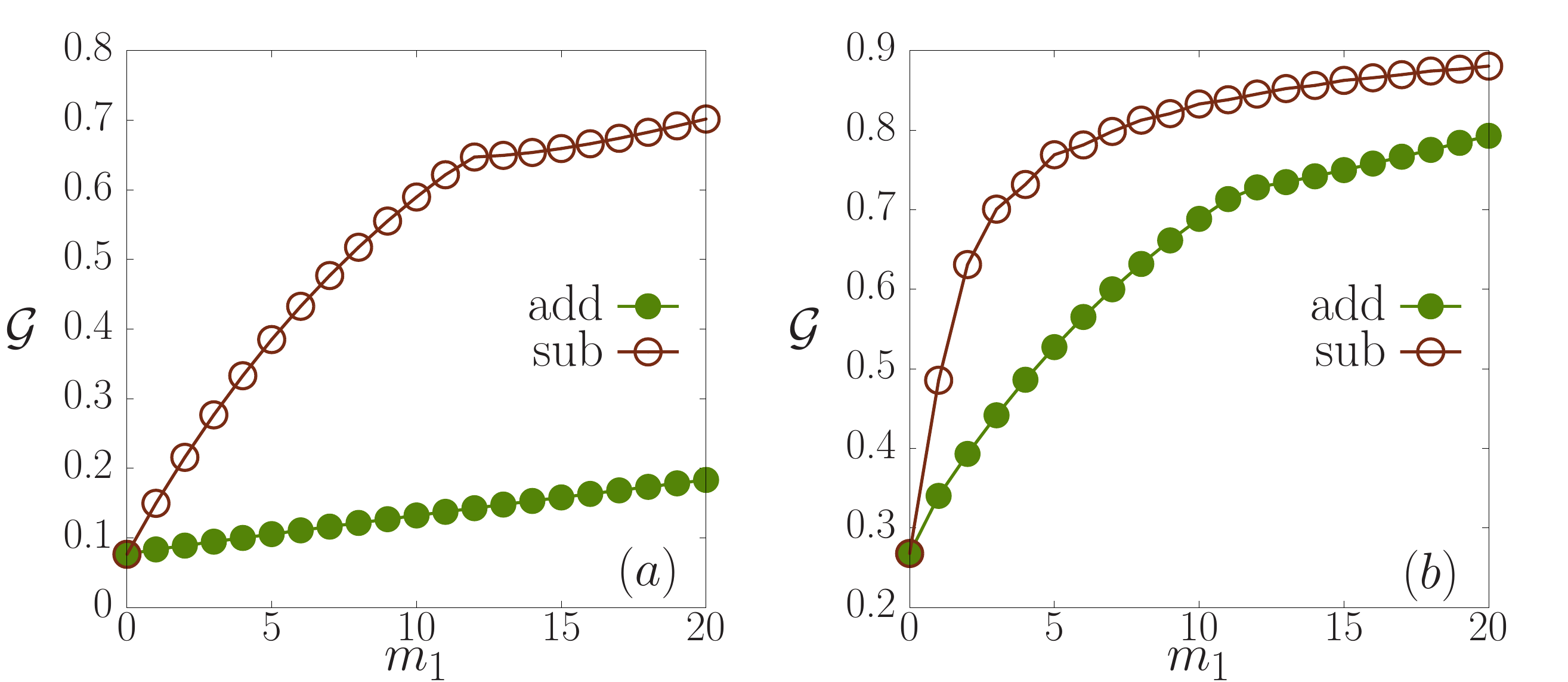}
\caption{
(Color online) GGM  against added (solid circles)  and subtracted (hollow circles) photons   from the first mode, \(m_1\). In (a), we set
the squeezing strength, r = 0.4, while we choose r = 0.8 in (b). In both the cases,  subtraction yields higher value of GGM than that of the addition.
%Plot of enhancement of GGM upon adding or subtracting photons from a single mode with respect to $m_1$. In (a) we set the squeezing strength $r=0.4$, while we set $r = 0.8$ in (b). In both cases subtraction yields greater enhancement of GGM.
}\label{fig:FMSV-mode1}
\end{figure}

\subsection{GGM of non-Gaussian states originated from FMSV }
\label{sec:non-G-GGM}
We now investigate the multimode entanglement content of a four-mode non-Gaussian state which emerges by adding and subtracting photons from the modes of the FMSV state given in Eq. (\ref{eq:FMSV}), and our aim is to find whether the increase of multimode entanglement with respect to addition and subtraction of photons as seen in three-mode non-Gaussian states can also persist in the four-mode case also. 
%The characteristics of the FMSV state are quite rich in comparison to the BSV state, which is quite clear from the expression of both the state. Moreover, the fock basis representation of the FMSV state is quite complicated compared to BSV state, as all possible bipartition of BSV state reduces to a diagonal state, but there are some bipartion in FMSV state which are not even block diagonal. 
The photon-added and -subtracted FMSV states read respectively  as 
\begin{widetext}
\begin{eqnarray}
|\psi^{add\{m_i\}}_{FMSV}\rangle = \frac{1}{N^{add}}  \sum_{n = 0}^\infty \sum_{r_1 = 0}^n \sum_{r_2 = 0}^n (\frac 12 \tanh r)^n \sqrt{{n}\choose{r_1}} \sqrt{{n}\choose{r_2}}  \sqrt{\frac{(n -r_1 + m_1)!}{(n - r_1)!}}  \sqrt{\frac{(n -r_2 + m_2)!}{(n - r_2)!}}  \nonumber \\
\sqrt{\frac{(r_1 + m_3)!}{r_1!}} \sqrt{\frac{(r_2 + m_4)!}{r_2!}} |n - r_1 + m_1\rangle |n - r_2 + m_2\rangle | r_1 + m_3\rangle | r_2 + m_4\rangle, \label{eq:add-FMSV}\\
|\psi^{sub\{m_i\}}_{FMSV}\rangle = \frac{1}{N^{sub}}  \sum_{n = M}^\infty \sum_{r_1 = m_3}^{n - m_1} \sum_{r_2 = m_4}^{n - m_2} (\frac 12 \tanh r)^n \sqrt{{n}\choose{r_1}} \sqrt{{n}\choose{r_2}}  \sqrt{\frac{(n -r_1 )!}{(n - r_1 - m_1)!}}  \sqrt{\frac{(n -r_2 )!}{(n - r_2 - m_2)!}}  \nonumber \\
\sqrt{\frac{r_1!}{(r_1 - m_3)!}} \sqrt{\frac{r_2!}{(r_2 - m_4)!}} |n - r_1 - m_1\rangle |n - r_2 - m_2\rangle | r_1 - m_3\rangle | r_2 - m_4\rangle, \label{eq:sub-FMSV}
\end{eqnarray}
\end{widetext}
where we use the convention that $m_i$ number of photons are added and subtracted in or from the mode $i$ $(i= 1, 2, 3, 4)$. $N^{add}$ and $N^{sub}$ are the respective normalization constants and $M = \max\{m_1 + m_3, m_2 + m_4\}$. Below, we catalogue the response of GGM to the operations of photon addition and subtraction:

\begin{enumerate}
\item \emph{Single-mode operations:} We find the pattern of GGM when we add or subtract photons from any one of the four modes, of the FMSV state.
%and investigate the genuinely multimode entanglement with the increasing number of photons added or subtracted. 
Note that in case of  single-mode photon operation, behavior of GGM remains independent of the choice of the mode taken for the operation. This is so because of the translational invariance of the Gaussian FMSV state.
Without loss of generality, we choose the first mode for the operation, i.e., $\mathcal{G}$ is computed by varying $m_1$,  and $m_2 = m_3 = m_4 = 0$.
 %any one among the four modes  are equally likely. In our calculation, we choose mode one for photon operation, i.e., $m_1$ will vary and $m_2 = m_3 = m_4 = 0$. 
%The plot of GGM for both the non-Gaussian FMSV state i.e., 
The plot of ${\cal G}(|\psi^{add\{m_1\}}_{FMSV}\rangle)$ and ${\cal G}(|\psi^{sub\{m_1\}}_{FMSV}\rangle)$,
 with respect to $m_1$ is depicted in Fig. \ref{fig:FMSV-mode1}, for two different values of the squeezing parameter $r = 0.4$ (in Fig. \ref{fig:FMSV-mode1}(a)) and $r = 0.8$ (in Fig. \ref{fig:FMSV-mode1}(b)).
From this figure, it is prominent that  ${\cal G}(|\psi^{sub\{m_1\}}_{FMSV}\rangle) \geq {\cal G}(|\psi^{add\{m_1\}}_{FMSV}\rangle)$ for fixed values of $m_1$, i.e., non-Gaussian state obtained by photon subtraction posses more genuine multimode entanglement than its photon-added counterpart, although both of them possess high amount of GGM compared to the corresponding Gaussian state (see $m_1 = 0$ point in Fig. \ref{fig:FMSV-mode1}). However, with the increase of the squeezing parameter, $r$ of the initial Gaussian state, 
the difference of GGM between the photon-subtracted and -added states diminishes due to the substantial increase of multimode entanglement in the photon-added state.

\item \emph{Two-mode operations:} Let us now observe whether such increment of multimode entanglement over the Gaussian states occurs even when photon-addition (-subtraction) are performed in two modes. In this situation, we consider two scenarios -- 
\begin{enumerate}
\item unconstrained (independent) mode operations in which photons are added (subtracted) in (from) two modes independently.

\item constrained mode operations where the total number of photons added (subtracted) in two modes are fixed. 
\end{enumerate}

By using the symmetry of the FMSV state,  we obtain that
there are only two  choices of operations possible: (i) operations in the alternate modes and (ii) operations involving the adjacent modes. Therefore, all two-mode operations can be segregated into these two bins. Without any loss of generality, we refer to operations in modes $1$ and $2$ as adjacent and that between modes $1$ and $3$ being alternate. 
In our case, we fix the number of photons added or subtracted from mode $1$ to be $m_1$, while that from modes $2$ and $3$ are both fixed to $n$, i.e., $m_2 = m_3 = n$, and the corresponding states would be denoted by  $|\psi^{add\{m_1,m_2=n\}}_{FMSV}\rangle$ and $|\psi^{sub\{m_1,m_2=n\}}_{FMSV}\rangle$ respectively.  
%We would refer to that state as $|\psi^{add\{m,0,n,0\}}_{FMSV}\rangle$, while the same in adjacent modes would be indicated as $|\psi^{add\{m,n,0,0\}}_{FMSV}\rangle$. The set of $4$ numbers in the superscript denotes the number of photons added in modes $1, ~2, ~3$ and $4$ respectively.  Same convention applies for photon subtracted state as well.
\begin{figure}
\includegraphics[width=\linewidth]{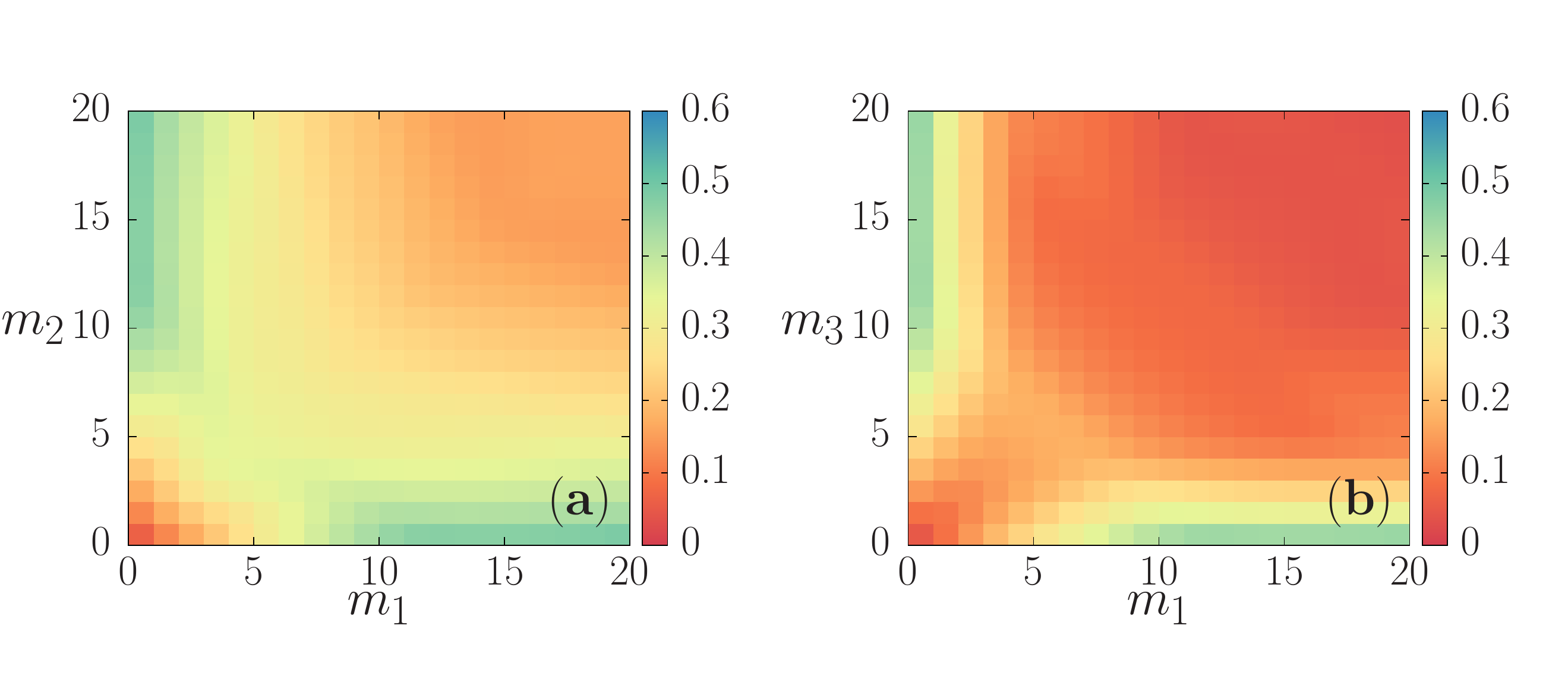}
\caption{ Difference between GGM of the photon-subtracted and -added states against \((m_1, m_2)\)-plane  in (a) and \((m_1, m_3)\)-plane in (b). 
In particular, we study the behavior of   ${\cal G}(|\psi^{sub\{m_1,m_2\}}_{FMSV}\rangle)- {\cal G}(|\psi^{add\{m_1,m_2\}}_{FMSV}\rangle)$, and ${\cal G}(|\psi^{sub\{m_1,m_3\}}_{FMSV}\rangle)- {\cal G}(|\psi^{add\{m_1,m_3\}}_{FMSV}\rangle)$. Positive values guarantee that subtraction is always better than that of the addition. 
 We set \(r =0.4\). 
%Plot of $\delta^{\{m_i\}}_{FMSV}$ with respect to the added and subtracted photon numbers in the adjacent and alternate  modes. In a) photon numbers independently varies between the adjacent modes, and no photons are added or subtracted in or from the remaining modes. Whereas in b) alternate modes has been chosen for operations and no subtraction and addition of photons happen in the remaining modes. Squeezing parameter $r$ has been chosen to be  $0.4$.
}
\label{fig:FMSV-twomode}
\end{figure}

$\bullet$ For unconstrained two-mode operations, which include both adjacent and alternate mode operations, photon-subtraction yields higher GGM compared to photon-addition. We observe this by considering the difference in GGM values of the two states, one of which is obtained by subtracting photons from the FMSV state, in the adjacent (alternate) modes, while the other by adding the same  number of photons to the FMSV state in the adjacent (alternate) modes. Specifically, we 
observe that the positive value of the quantity ${\cal G}(|\psi^{sub\{m_1,m_2\}}_{FMSV}\rangle)- {\cal G}(|\psi^{add\{m_1,m_2\}}_{FMSV}\rangle)$ in Fig. \ref{fig:FMSV-twomode} (a), and ${\cal G}(|\psi^{sub\{m_1,m_3\}}_{FMSV}\rangle)- {\cal G}(|\psi^{add\{m_1,m_3\}}_{FMSV}\rangle)$ in Fig. \ref{fig:FMSV-twomode} (b) ensure the superiority of photon-subtraction in these cases.

Let us now find whether operations in adjacent modes is more beneficial than that of the alternate modes. To address this question, 
 we now restrict ourselves to photon addition. In  this case,  the operations in alternate modes gives higher value of GGM than that of the adjacent modes. We demonstrate this by plotting ${\cal G}(|\psi^{add\{m_1,m_3 =n\}}_{FMSV}\rangle)- {\cal G}(|\psi^{add\{m_1,m_2 =n\}}_{FMSV}\rangle)$ in Fig.  \ref{fig:mode-compare}(a). The positive value in the entire ranges of $(m_1,n)$-plane certifies supremacy of alternate operations in case of photon addition. Interestingly, for photon subtraction, we observe a completely opposite effect, i.e.,  the adjacent operations lead to  higher GGM values, as indicated by negative values of  ${\cal G}(|\psi^{sub\{m_1,m_3 =n\}}_{FMSV}\rangle)- {\cal G}(|\psi^{sub\{m_1,m_2 =n\}}_{FMSV}\rangle)$ (see Fig. \ref{fig:mode-compare} (b) with  $m_1$ and $n$).
\begin{figure}[h]
\includegraphics[width=\linewidth]{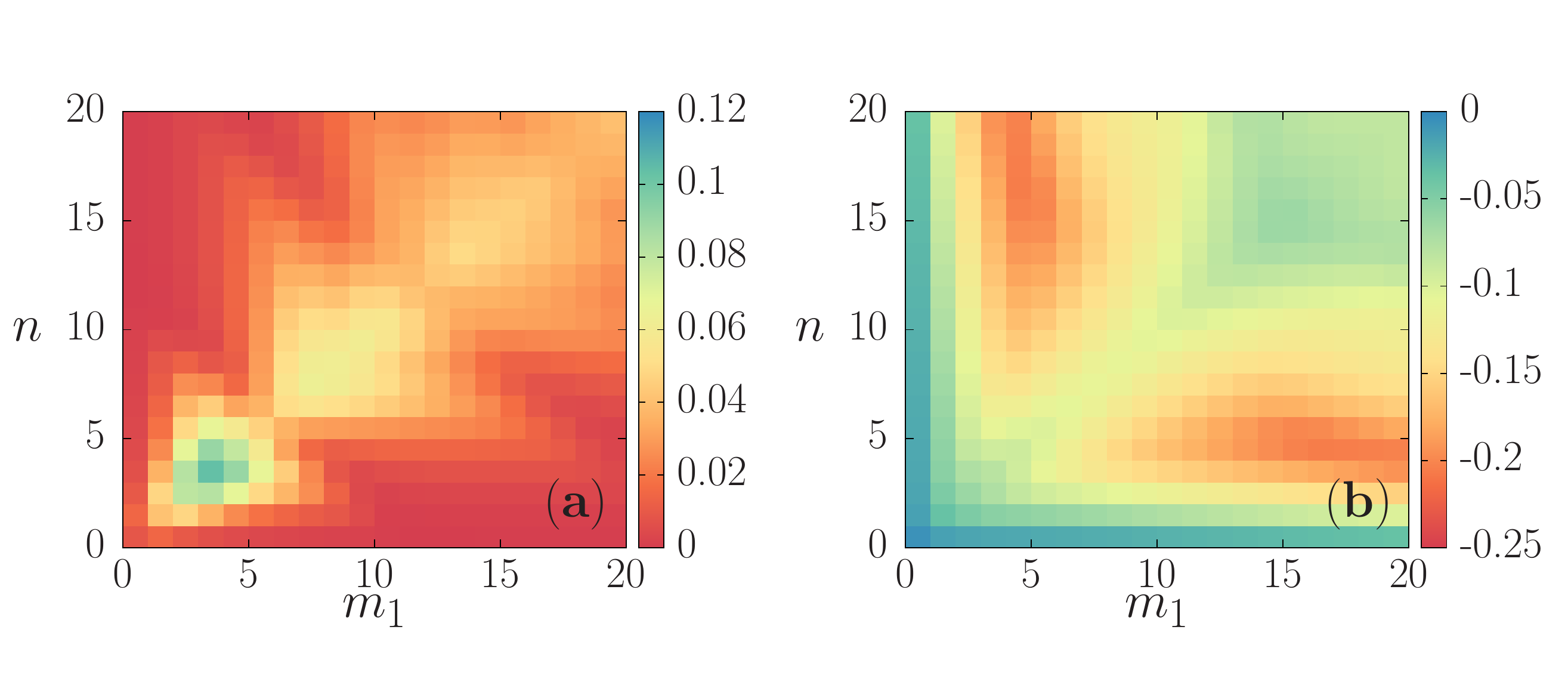}
\caption{(Color Online) Map of GGM in the \((m_1, n)\)-plane. (a) \(m_1\) and \(m_{2(3)}=n\) photons are added in modes. To understand whether adding
photons in adjacent modes can generate more multimode entanglement than  that of the alternate mode, we plot
\({\cal G}(|\psi^{add\{m_1,m_3 =n\}}_{FMSV}\rangle)) - {\cal G}(|\psi^{add\{m_1,m_2 =n\}}_{FMSV}\rangle) \). We find that the adding photons in the
alternate modes is beneficial than that of the adjacent modes. In (b),  subtraction of photons is considered. Negative value of the quantity, \({\cal G}(|\psi^{sub\{m_1,m_3 =n\}}_{FMSV}\rangle)) - {\cal G}(|\psi^{sub\{m_1,m_2 =n\}}_{FMSV}\rangle) \), indicates
that the picture is opposite than the addition, thereby specifying that   subtracting photons in alternate mode is an advantageous enterprise.  Here r is fixed to 0.4.
%The enhancements in GGM are compared for adjacent and alternate two mode operations. For subtraction (a), adjacent mode operations yield greater enhancement while for addition (b), the alternate mode operation turn out to be better. $r$ is fixed to $0.4$.
}
\label{fig:mode-compare}
\end{figure}

\begin{figure}[h]
\includegraphics[width=\linewidth]{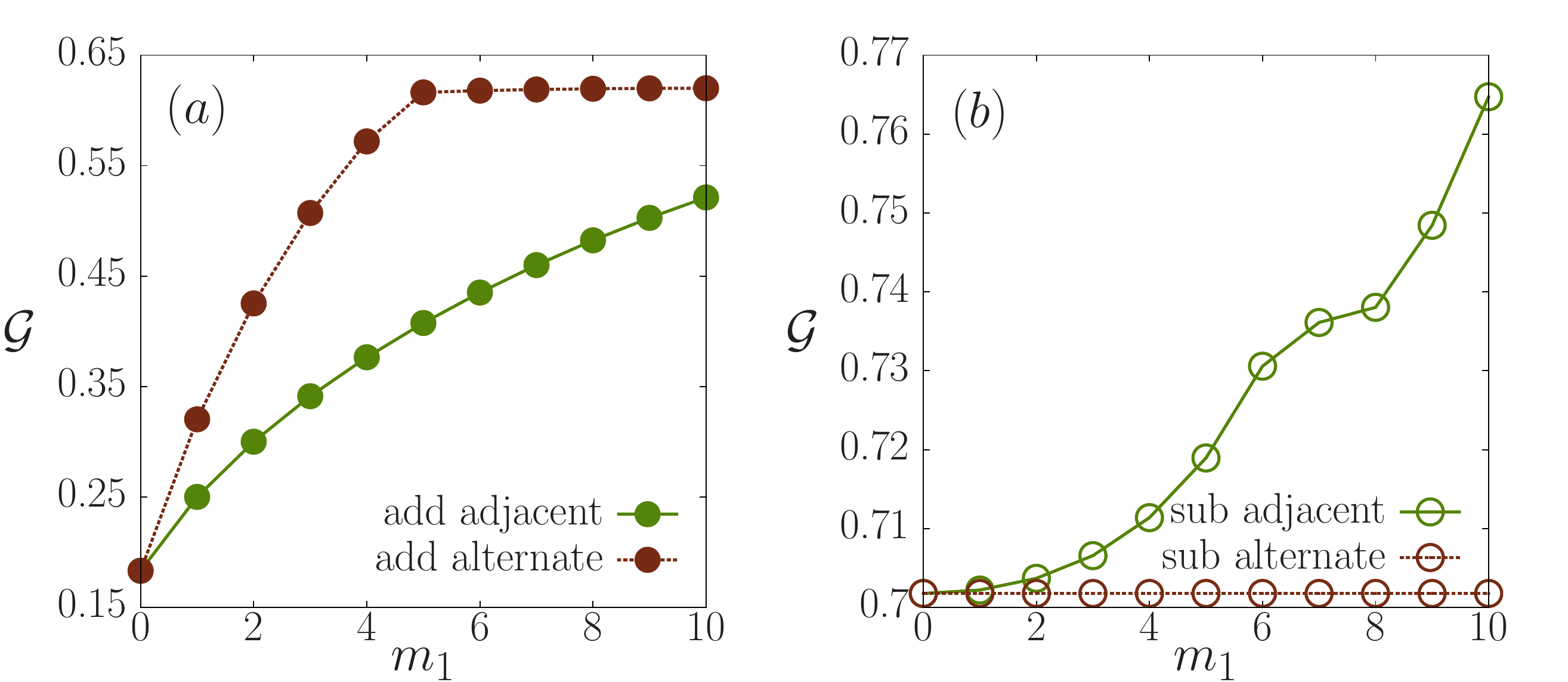}
\caption{(
Color online) GGM vs. \(m_1\) with \(m_1 + m_2 = 20\). (a) GGM is measured after addition of \(20\) photons in the adjacent, \(m_1 \) and \(m_2\) (solid circles) and in the alternate, \(m_1\) and \(m_3\), (solid squares) modes. (b) GGM is plotted after subtracting photons from the adjacent (hollow circles) and the alternate (hollow squares) modes. Here \(r  = 0.4\). When photons are subtracted from the alternate modes,  GGM remains constant with \(m_1\), we call such a phenomena as freezing, as proven in Sec. IVC. 
%Variation of GGM with (a) addition and (b) subtraction of $20$ photons from adjacent or alternate modes. In case of addition, distributing the photons in the alternate modes yield greater enhancement of GGM. For subtraction from alternate modes, we observe a freezing behavior in GGM. Here $r$ is fixed to $0.4$.
}
\label{fig:constraint-sub-add}
\end{figure}

$\bullet$ For constrained photon operations, we fix the total number of photons added (subtracted) in either the adjacent  or the alternate modes. Here we fix $m_1 + m_{2(3)} = 20$ and in Fig. \ref{fig:constraint-sub-add}(a), we consider the photon addition operation, and plot ${\cal G}(|\psi^{add\{m_1,m_2 = 20- m_1\}}_{FMSV}\rangle)$, indicated by the curve with solid green circles, and ${\cal G}(|\psi^{add\{m_1,m_3 = 20- m_1\}}_{FMSV}\rangle)$, represented by the curve with solid brown circles, with respect to $m_1$.  In this situation,  the alternate modes lead to higher elevation of GGM than that obtained for operations in the adjacent modes.

This trend is qualitatively different when we consider  photon subtraction (as  depicted in Fig.  \ref{fig:constraint-sub-add}(b)).  We find that ${\cal G}(|\psi^{sub\{m_1,m_2 = 20- m_1\}}_{FMSV}\rangle)$ increases with increasing $m_1$ ( hollow green circles in Fig. \ref{fig:constraint-sub-add}(b)), while ${\cal G}(|\psi^{sub\{m_1,m_3 = 20- m_1\}}_{FMSV}\rangle)$ remains constant with the  increase of $m_1$ (hollow red circles in Fig. \ref{fig:constraint-sub-add}(b)). We describe this feature as freezing of GGM and will provide an analytical analysis of the same in Sec. \ref{sec:freezing}. 
Note that in case of photon-addition, \(\mathcal{G}\) slowly varies with high values  of \(m_1\) ( as depicted in  \ref{fig:constraint-sub-add}(a)) and so it is not similar to the freezing phenomena observed. 

%the apparent straight line in \ref{fig:constraint-sub-add}(a) solely depend on the scaling of the plot, as there is actually a slight variation in the GGM values with respect to $m_1$. So this is not the same as freezing obtained during photon subtraction from alternate modes.

%
% 
%We observe that involving alternate modes gives better result, although the enhancement became almost constant after $m_1 = 5$. The behavior of constrained operation for photon subtraction has been depicted in Fig.  \ref{fig:constraint-sub-add}(b). Here we find that GGM freezes for alternate mode operations, but enhance quiet faster for adjacent mode operation. The red hollow circular line in \ref{fig:constraint-sub-add}(b) is constant with the increasing number of $m_1$, which is slightly equivalent with the same operation in \ref{fig:constraint-sub-add}(a). But both the behavior are completely different. The apparent straight line in \ref{fig:constraint-sub-add}(a) solely depend on the scaling of the plot, there is a slight variation in the value in the 3rd decimal point. But the straight line in \ref{fig:constraint-sub-add}(b) is constant with $m_1$, we coined the term of freezing of GGM which we will describe in  section \ref{sec:freezing}.
\begin{figure}[h]
\includegraphics[width=\linewidth]{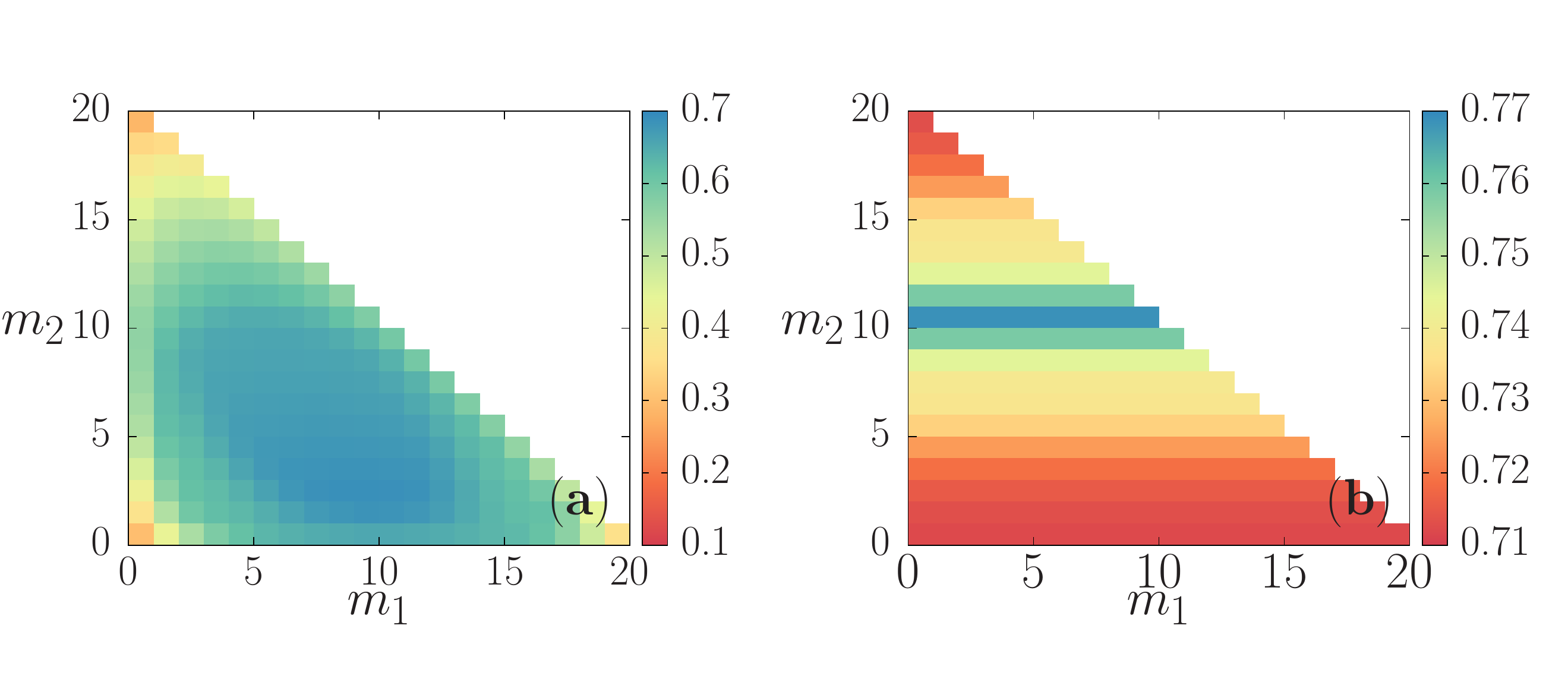}
\caption{
GGM with respect to \(m_1\) and \(m_2\), when total number of photons added (subtracted) in three modes is fixed to \(20\) i.e., \(m_1 + m_2 + m_3 = 20\). (a) corresponds to   addition of photons. GGM is generated when equal number of photons are added in all the three modes.  In (b), subtraction of photons is plotted with \(m_1\) and \(m_2\).  For a fixed \(m_2\), we find that the GGM remains constant with \(m_1 + m_3\),  the freezing phenomena, similar
to the one obtained in Fig. 7(b) (see also Theorem 2 in Sec. IVC). Here
\(r =0.4\).
%Enhancement in GGM for $3$-mode operations. Total number of photons is fixed to $20$. For (a) addition, the most uniform distribution produces the highest GGM, which is not the case in (b) subtraction due to the freezing phenomena.
}
\label{fig:three-mode}
\end{figure}

\item \emph{Multimode operations:} 
Let us now consider a situation where three modes are involved in the photon-addition or -subtraction scheme. 
%Here we restrict ourselves to constrained operations, especially we
Let us  fix $m_1 + m_2 + m_3 = 20$, where $m_1$, $m_2$, and $m_3$ denote the numbers of photons added (subtracted) in the first, second and third mode respectively. As depicted in Fig. \ref{fig:three-mode}, the characteristics of GGM in this scenario is qualitatively  different for addition and subtraction. Specifically, in case of photon-addition, we find that the GGM reaches its maximal value when all the three modes possess almost equal number of photons. However, in case of subtraction, maximal value of GGM is independent of $m_1$, i.e., we obtain freezing along $m_1$ which can also be explained analytically by the Theorem below.

%We also made another investigation, which involved three modes, but to visualize the behavior of GGM we made a constrained operation. We choose modes one, two and three for addition and subtraction operations and fix the total number of added and subtracted photons to $m_1 + m_2 + m_3 = 20$. The GGM has been plotted in Fig. \ref{fig:three-mode}, where we vary $m_1$ along $x$-axis and $m_2$ along the $y$-axis, Fig. \ref{fig:three-mode}(a) is for photon addition whereas Fig. \ref{fig:three-mode}(b) is for photon subtraction. For adding photons in all the three modes we have found that maximum enhancement occur when photons are distributed among all the three modes. But for subtraction, maximum appears at $m_2 = 10$, and it is independent of $m_1$, i.e., freezing along $m_1$.

\end{enumerate}

\subsection{Freezing of GGM}\label{sec:freezing}

%Moreover, we will also prove analytically that for the constraint mode operation between mode one and three of the FMSV state, the GGM freezes with respect to the subtracted photon number in mode one, which is not the case for photon addition.

Let us show that the inherent form of the photon-subtracted state involving alternated modes ensures the freezing feature of GGM with the number of photons subtracted. Specifically, we obtain the following: 
%The feature of GGM obtained in Fig. \ref{fig:constraint-sub-add} (b) can be proven analytically in the following theorem:
\begin{theorem}
In case of constant total number of photons subtracted from two alternate modes, 
the genuine multimode entanglement of a non-Gaussian FMSV state is independent of the number of photons subtracted from any one of those modes.
\label{th:thm2}
\end{theorem}
\begin{proof}
Without loss of generality, we assume that 
%Before we begin the proof, we recall the terminology of the alternate modes for the FMSV state given in Eq. (\ref{eq:FMSV}). We consider that modes $1$ and $3$ are alternate to each other similarly modes $2$ and $4$. For constrained photon subtraction, here 
a fixed total number photons are subtracted from modes $1$ and $3$, i.e., $m_1 + m_3 = $ constant, and no photons are added or subtracted from second and fourth modes.  To prove our claim, we will now show that the coefficients of the photon-subtracted state, given in Eq. (\ref{eq:sub-FMSV}) is only a function of the total number of photons subtracted from the alternate modes, $m_1 + m_3 = M$, say. By putting $m_2 = m_4 = 0$, Eq. (\ref{eq:sub-FMSV}) reduces to
\begin{widetext}
\begin{eqnarray}
|\psi^{sub}_{FMSV}\rangle &=& \frac{1}{N^{sub}}  \sum_{n = M}^\infty \sum_{r_1 = m_3}^{n - m_1} \sum_{r_2 = 0}^{n} (\frac 12 \tanh r)^n \sqrt{{n}\choose{r_1}} \sqrt{{n}\choose{r_2}}  \sqrt{\frac{(n -r_1 )!}{(n - r_1 - m_1)!}}  \sqrt{\frac{r_1!}{(r_1 - m_3)!}}  \nonumber \\
&& |n - r_1 - m_1\rangle |n - r_2\rangle | r_1 - m_3\rangle | r_2\rangle,  \\
&=& \frac{1}{N^{sub}}  \sum_{n = 0}^\infty \sum_{r_1 = 0}^{n} \sum_{r_2 = 0}^{n} (\frac 12 \tanh r)^{n + M} \sqrt{{n + M}\choose{r_1 + m_3}} \sqrt{{n}\choose{r_2}}  \sqrt{\frac{(n + m_1 - r_1  )!}{(n - r_1)!}}  \sqrt{\frac{(r_1 + m_3)!}{r_1!}}  \nonumber \\
&& |n - r_1\rangle |n - r_2\rangle | r_1 \rangle | r_2\rangle, \\
&=& \frac{1}{N^{sub}}  \sum_{n = 0}^\infty \sum_{r_1 = 0}^{n} \sum_{r_2 = 0}^{n} (\frac 12 \tanh r)^{n + M}  \sqrt{{n}\choose{r_2}}  \sqrt{\frac{(n + M )!}{(n - r_1)! ~ r_1!}}   
 ~|n - r_1\rangle |n - r_2\rangle | r_1 \rangle | r_2\rangle. 
 \label{eq:fmsv-freeze}
\end{eqnarray}
\end{widetext}
In the second equality, we make the change of variables, namely $r_1 \rightarrow r_1 + m_3$ and $n \rightarrow n + M$. From Eq. \eqref{eq:fmsv-freeze}, it is clear that the above state is only function of $M$, and hence independent of $m_1$ and $m_3$ individually, but only dependent on their sum.
\end{proof}

\emph{Remark: } Note that the above theorem is proven for $m_1 + m_3 = M$ and $m_2 = 0$. Fig. \ref{fig:three-mode} (b) shows that Theorem \ref{th:thm2} also holds for the scenario when $m_1 + m_3 = M$ and $m_2 = c$, where $c$ is a fixed integer.

\subsection{Analyzing GGM enhancement via non-Gaussianity}
\label{sec:non-gau}
\begin{figure}
\includegraphics[width=\linewidth]{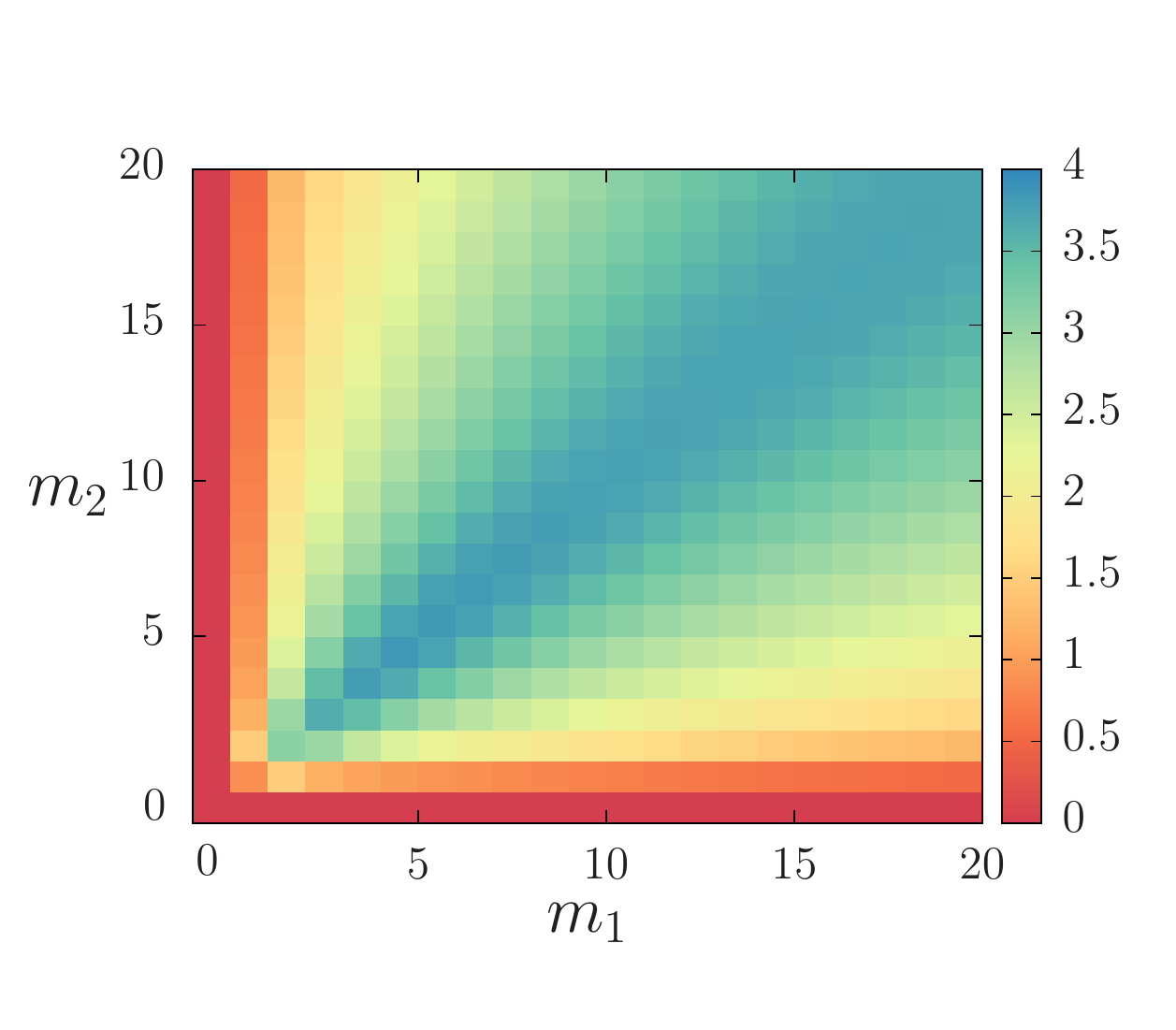}
\caption{Difference between non-Gaussianity, $\delta_{NG}(|\psi^{add\lbrace m_1,m_2 \rbrace}_{FMSV}\rangle) - \delta_{NG}(|\psi^{sub\lbrace m_1,m_2 \rbrace}_{FMSV}\rangle)$, of the photon-added and -subtracted states against $(m_1 , m_ 2 )$-plane. We observe that for single mode operations, addition and subtraction of photons lead to the same amount of non-Gaussianity, while for two mode operations, photon addition leads to greater non-Gaussianity. The squeezing parameter is fixed to 0.4.}
\label{fig:ng}
\end{figure}

\begin{table}[h]
\begin{center}
\begin{tabular}{|P{1.2cm}|P{1.7cm}|P{1.7cm}|P{1.7cm}|P{1.7cm}|}
\hline
    $m_1,m_2$ &  $\delta_{NG}^{add}$ &  $f_{\mathcal{G}}^{\text{add}}$ &  $\delta_{NG}^{sub}$ &  $f_{\mathcal{G}}^{sub}$\\
    \hline
    \hline
    $2,0$ & $2.7548$ & $0.1639$ & $2.7548$ & $1.8234$ \\
\hline
$5,0$ & $3.9001$ & $0.3792$ & $3.9001$ & $4.0378$ \\
\hline
$10,0$ & $4.8344$ & $0.7293$ & $4.8344$ & $6.719$ \\
\hline
$2,1$ & $4.5661$ & $0.7584$ & $2.0653$ & $2.9408$ \\
\hline
$5,1$ & $5.5072$ & $1.0780$ & $3.6208$ & $4.7684$ \\
\hline
$10,1$ & $6.1742$ & $1.5493$ & $4.6930$ & $7.1720$ \\
\hline
\end{tabular}
\end{center}
\caption{Comparative study of fractional enhancement of GGM from FMSV value and non-Gaussianity on addition or subtraction of photons from a single and two modes of the FMSV state. Note that $\delta_{NG}^{add/sub} \equiv \delta_{NG}(|\psi^{add/sub\lbrace m_1,m_2 \rbrace}_{FMSV}\rangle)$. The squeezing parameter is fixed to 0.4.}
\label{table:ng-ggm}
\end{table}

In this section, we will try to figure out the departure of a photon added and subtracted states from the initial Gaussian FMSV state in terms of the relative entropy based measures of non-Gaussianity. Our main motive is to answer the following question: \\
{\it Is there any connection between non-Gaussianity and GGM?}

\noindent To address the above question, whether the GGM amplification obtained via addition and subtraction of photons has any connection with the non-Gaussianity, we first restate the definition of the measure.  
%We argue that the enhancement of GGM  can be attributed to the increase of non-Gaussianity by the de-Gaussification via photon addition (subtraction). 
Specifically, the measure non-Gaussianity of a given state, $\rho$, in a CV system is defined in terms of relative entropy distance, as \cite{ng1, ng2, ng3}
\begin{equation}
\delta_{NG}(\rho) = \min_{\tau_G}  S(\rho||\tau_G),
\end{equation}
where the minimization is taken over all possible Gaussian states $\tau_G$, lying in the same Hilbert space of $\rho$.
The Gaussian state, which achieve the minimum possess  the same displacement vector and covariance matrix as $\rho$, lets consider it as $\rho_G$, hence
%by considering the departure of a given  from a Gaussian state  given by \cite{ng1, ng2, ng3}.
\begin{eqnarray}
\delta_{NG}(\rho) = S(\rho||\rho_G) = S(\rho)-S(\rho_G),
\end{eqnarray}
where $S(\sigma) = -\text{tr}~\sigma \log_2 \sigma$.
%, and $\rho_G$ is the Gaussian state with  Note that $\rho_G$ is the closest Gaussian state to $\rho$.
We find that $\delta_{NG}$ increases both for photon-added and -subtracted states.  Therefore, the enhancement of GGM seems to be consistent with  the increase of non-Gaussianity as  the number of  photons  added (subtracted) increases, and thus can be plausibly argued  as the \emph{physical} reason behind the observed improvements.
However, a more detailed analysis reveals that in most cases (see Fig. \ref{fig:ng}), photon addition leads to much faster de-Gaussification compared to photon subtraction although the GGM content of  the photon-subtracted state is higher than that of the photon-added ones, see Figs. \ref{fig:FMSV-mode1} and \ref{fig:FMSV-twomode}. So, non-Gaussianity cannot explain the difference in response of photon addition and subtraction.  We would like to mention here that the physical reason for the different behaviour obtained for photon addition and subtraction remained unanswered in the earlier works as well \cite{TMSV, FMSV, enhance1, bell, enhance3, enhance2}.
 
Nevertheless,  a quantitative analysis is carried out  between non-Gaussianity and fractional enhancement of GGM for the photons-added (-subtracted) states where the later one is defined as 
\begin{eqnarray}
f_G^{add/sub} = \frac{\mathcal{G}(|\psi^{add/sub\{m_1,m_2\}}_{FMSV}\rangle) - \mathcal{G}(|FMSV\rangle)}{\mathcal{G}(|FMSV\rangle)}.
\end{eqnarray}
It allows us to track how fast the GGM grows on addition or subtraction of photons, see Table. \ref{table:ng-ggm}.
%We compute these quantities and present our findings in Fig. \ref{fig:ng} \textcolor{red}{fig 9 e to eq. (44) plot korini} and Table. \ref{table:ng-ggm}.

To summarize, Fig. \ref{fig:ng} and Table. \ref{table:ng-ggm} reveal three distinct features --(1) when photons are added (subtracted) in a single mode, one gets the same amount of non-Gaussianity; (2) for two mode operations, photon addition always leads to higher non-Gaussianity compared to the photon subtraction;  (3)   in both the cases, GGM content in photon-subtracted states from FMSV is higher than that of the photon-added states, thereby showing  \(f_G^{sub} > f_G^{add}\).  Therefore, non-Gaussianity cannot conclusively discriminate photon addition and subtraction according to their   genuine multomide entanglement.
%To summarize, our analysis reveals that although non-Gaussianity can shed some light on the overall GGM enhancement on de-Gaussification, it is incapable of conclusively distinguishing the GGM enhancing capability of photon addition and subtraction procedures. 
%In quest of this answer, we have also investigated
 Notice also that  the difference between entanglement values of photon-added and -subtracted states cannot be explained by other physical quantities  like departure of the photon- added (-subtracted)  state from the initial Gaussian FMSV state,  the fidelity of photon-added (-subtracted)  state with a FMSV state having the same value of GGM  \cite{Bartley-2015}.
We believe that a complete satisfactory and ubiquitous physical reasoning may depend on the structure of the state itself, on which mode(s) the photonic operations are performed, and on the quantity of interest.

\section{Conclusion}
\label{sec:con}
The generalized geometric measure (GGM), which owes its origin to the geometry of quantum state space, has established itself as a computationally efficient quantifier of genuine multiparty entanglement of pure states in finite dimensional systems. Based on Schmidt decomposition,
% for infinite dimensional systems,
 its applicability can also be extended to compute  genuine multimode entanglement  between the various modes of a multimode pure state consisting of infinite dimensional subsystems. 
%Such a distance-based definition can also faithfully quantify genuine multimode entanglement in infinite dimensional systems. We proved that the measure for pure multimode state can be computed by considering Schmidt decomposition in different modal bipartitions.

%Formally, the canonical formula of GGM is expressed as the difference between the algebraically obtainable maxima of unity from the maximal eigenvalue of the relevant reduced states. 
%For infinite dimensional systems, apparently computation of GGM looks
Nevertheless, such simplified formula can also be  intractable due to the infinite dimensional structures inherent in these systems. %that have to be handled. 
However, we proved that for Gaussian multimode pure states, GGM can be expressed  in terms of the symplectic eigenvalues of the covariance matrices obtained from the various relevant reduced states. Essentially, we are able to compute the coverted maximal eigenvalue in terms of the symplectic invariants and hence this simplification. To illustrate the functionality of the formula, we computed GGM of some prototypical three- and four-mode Gaussian states. Note that, the formula for GGM is efficiently scalable, can be employed to evaluate GGM of Gaussian states with higher number of modes, and its pertinence is not merely restricted to three- and four-mode states. 
Furthermore, the symplectic eigenvalues are evaluated from the covariance matrices which are in turn composed of quadrature correlations, the typical quantities extracted in experiments. This experimental friendliness adds another point of merit to our work.
%} ``or" \textcolor{blue}{This feature makes our work important from an experimental point of view as well.}

When one proceeds beyond the Gaussian paradigm, such simplified evaluation of GGM is not possible.
%, and the canonical formula of GGM is the only tool that remains at our disposal. 
However, for some non-Gaussian states (like photon-added and -subtracted states), symmetries in the structure of the states enables to find maximal eigenvalue of all reduced states.
% finding the task of the maximal eigenvalue tractable.  
De-Gaussification via photon addition and subtraction displays substantial enhancement of GGM from the Gaussian value. We performed a comparative study in the increase of genuine multimode entanglement induced by addition and subtraction in the three- and four-mode scenarios and observe some novel features. A careful analysis also  revealed that although the increase of genuine multimode entanglement content of photon-added (-subtracted) states over the four mode squeezed vacuum states can be answered by the increase of the distance-based non-Gaussianity measure, the superiority in terms of multimode entanglement value obtained for  photon subtraction compared to photon-addition cannot be explained by the non-Gaussianity measure. 

To summarize, our work sheds  light on the quantification of genuine multimode entanglement in continuous variable systems. We believe, our work will be a stepping stone for further systematic analysis of multimode entanglement involving non-Gaussian states and mixed states.
%
%\emph{``If winter comes, can spring be far behind?'' } 
%
%\hfill --  Ode to the West Wind, P. B. Shelley. 
%

\acknowledgements
%The authors thank Ujjwal Sen for insightful discussions.
TD thanks Prof. Marek Zukowski and Bianka Woloncewicz for insightful comments and discussion and acknowledge The ’International Centre for Theory of Quantum Technologies’ project (contract no. 2018/MAB/5).

\appendix
\section{Schmidt decomposition in CV system}
\label{app:app1}
In this section, we are trying to investigate that if in the continuous variable system, matrix diagonalization is possible, then the Schmidt decomposition is also possible.

\textbf{Statement:} Suppose $A$ is a $d$ dimensional normal matrix i.e., $AA^\dagger = A^\dagger A$. The Hermitian matrix $H^\dagger = H$ and unitary matrix $U^\dagger U = I$ are the example of normal matrices. There always exists, an unitary matrix $U$,  such that $U A U^\dagger = A_{diag}$, where $A_{diag}$ is a diagonal matrix, i.e, $A_{ij} = \delta_{ij} a_i$.
The above statement is also true in infinite dimensions provided the operator is compact and Hermitian.

%Now in case of  self addjoint operator $A$, there also exists unitary operator $U$, which diagonalize $A$, even it is an infinite dimensional.
%infinite dimensional hermitian operator $A$, if the operator is compact then also it can 
%we assume that the above statement is also true if we take $d \rightarrow \infty$, i.e., $U A U^\dagger = B_{dig}$ holds for infinite dimensional system.

\subsection*{Proof of Schmidt decomposition in infinite dimension}
Suppose $|\psi_{12}\rangle$, is a two-mode continuous variable system. $|\psi_{12}\rangle \in {\cal H}^1 \otimes {\cal H}^2$, ${\cal H}^1 ({\cal H}^2)$ are two infinite dimensional Hilbert spaces, with  countably infinite number of fock basis $\{|n\rangle\}$ and $\{|m\rangle\}$. Hence, it can be expanded  as
\begin{equation}
|\psi_{12}\rangle = \sum_{m,n =0}^\infty a_{m,n} |m,n\rangle,
\end{equation}
where $a_{m,n}$ is the complex coefficient, satisfy $\sum_{m,n = 0}^\infty |a_{m,n}|^2 = 1$. Let us consider a single-mode density matrix as 
\begin{eqnarray}
\rho_1 = \text{tr}_2 (|\psi_{12}\rangle\langle \psi_{12}|) &=& \sum_{m,m'}^\infty \Big(\sum_n a_{m,n}a^*_{m',n} \Big) |m\rangle\langle m'| \nonumber \\
&=& \sum_{m,m'}^\infty  C_{m,m'} |m\rangle\langle m'| %\sum_{\mu} \lambda_\mu |\mu\rangle\langle \mu|,
\end{eqnarray}
where $ C_{m,m'} = \sum_n^\infty a_{m,n}a^*_{m',n}$. It is clear that $C^\dagger = C$, i.e., $C$ is a Hermitian matrix. Moreover $C$ is also a compact matrix which guarantees that the spectral theorem can also be applied here \cite{schmidt2}.
 Hence we have an unitary $U$, such that $U C U^\dagger = \Lambda_{diag}$.
Hence, 
\begin{eqnarray}
\rho_1 =  \sum_{m,m'}  C_{m,m'} |m\rangle\langle m'| &=& \sum_{m,m', \mu} U_{m,\mu} \lambda_{\mu} U^*_{m',\mu} \lambda_{\mu} |m\rangle\langle m'| \nonumber \\ &=& \sum_{\mu} \lambda_\mu |\mu\rangle\langle \mu|,
\end{eqnarray}
where $|\mu\rangle =\sum_m U_{m,\mu} |m\rangle$.
 $\lambda_\mu$s are the eigenvalue of $\rho_1$ and $\{|\mu\rangle\}$ is the eigenbasis or the Schmidt basis.

Now we can always write the initial pure state as $|\psi_{12}\rangle = \sum_{\mu} |\mu\rangle_1 |\tilde{\mu}\rangle_2$, where $|\tilde{\mu}\rangle = {}_1\langle \mu|\psi_{12}\rangle$, clearly $|\tilde{\mu}\rangle $ is not a normalized state. To prove that this expression is the schmidt decomposition of $|\psi_{12}\rangle$, we need to prove $\{|\bar{\mu}\rangle\}$ is mutually orthogonal.
% consider $|\tilde{\mu}\rangle  = \xi_\mu |\bar{\mu}\rangle$, where the $\{|\bar{\mu}\rangle\}$ are the normalized state. Hence we have $|\psi_{12}\rangle = \sum_{\mu} \xi_\mu |\mu\rangle |\bar{\mu\rangle}$. 

\begin{eqnarray}
|\tilde{\mu}\rangle = {}_1\langle \mu|\psi_{12}\rangle &=& \Big( \sum_{m'} U^*_{m',\mu} {}_1\langle m'| \Big)\Big( \sum_{m,n =0}^\infty a_{m,n} |m,n\rangle\Big) \nonumber \\
 &=& \sum_{m,n} U^*_{m,\mu} a_{mn} |n\rangle_2
\end{eqnarray}

Now
\begin{eqnarray}
\langle\tilde{\nu}|\tilde{\mu}\rangle &=& \sum_{m',n'} U_{m',\nu} a^*_{m'n'} {}_2\langle n'| \Big(\sum_{m,n} U^*_{m,\mu} a_{mn} |n\rangle_2\Big) \nonumber \\
&=& \sum_{m',m} U_{m',\nu} \sum_n a^*_{m'n} a_{mn} U^*_{m,\mu} \nonumber \\
&=& \sum_{m',m} U_{m',\nu} C_{m,m'} U^*_{m,\mu} = (U^\dagger CU)_{\nu,\mu} = \delta_{\mu,\nu}\lambda_{\mu}, \nonumber \\
\end{eqnarray}

By choosing $|\tilde{\mu}\rangle = \sqrt{\lambda_\mu} |\mu\rangle$, we get the Schmidt decomposition of a pure state $|\psi_{12}\rangle$ in CV system as
\begin{equation}
|\psi_{12}\rangle = \sum_{\mu}^\infty \lambda_\mu |\mu,\mu\rangle.
\end{equation}

\section{Covariance matrix of the FMSV state}
\label{app:app2}
The FMSV state \cite{FMSV} possess the following covariance matrix: \vspace{0.2cm}\\

 $\Lambda^{\text{FMSV}} = $
\begin{eqnarray}
\frac{1}{2}
\begin{bmatrix}
\cosh^2 r ~\mathbb{I}_2 & \frac{1}{2}\sinh 2r ~\sigma_z & \sinh^2 r ~\mathbb{I}_2 & \frac{1}{2}\sinh 2r ~\sigma_z \\
\frac{1}{2}\sinh 2r ~\sigma_z & \cosh^2 r ~\mathbb{I}_2 & \frac{1}{2}\sinh 2r ~\sigma_z & \sinh^2 r ~\mathbb{I}_2 \\
\sinh^2 r ~\mathbb{I}_2 & \frac{1}{2}\sinh 2r ~\sigma_z & \cosh^2 r ~\mathbb{I}_2 & \frac{1}{2}\sinh 2r ~\sigma_z \\
\frac{1}{2}\sinh 2r ~\sigma_z & \sinh^2 r ~\mathbb{I}_2 & \frac{1}{2}\sinh 2r ~\sigma_z & \cosh^2 r ~\mathbb{I}_2 \nonumber 
\end{bmatrix},
\end{eqnarray}
where $\mathbb{I}_2 = $ diag$\{ 1,1\}$, and $\sigma_z = $diag$\{1,-1 \}$. All the single-mode reduced covariance matrices are identical and is given by, $\Lambda^{\text{FMSV}}_{\text{single}} =\frac12 \cosh^2 r ~\mathbb{I}_2$. The relevant two-mode reduced covariance matrices, corresponding to the adjacent and alternate modes are given by 
\begin{eqnarray}
\Lambda^{\text{FMSV}}_{\text{adjacent}} = \frac{1}{2}
\begin{bmatrix}
\cosh^2 r ~\mathbb{I}_2 & \frac{1}{2}\sinh 2r ~\sigma_z \\
\frac{1}{2}\sinh 2r ~\sigma_z & \cosh^2 r ~\mathbb{I}_2
\end{bmatrix},
\end{eqnarray}
and
\begin{eqnarray}
\Lambda^{\text{FMSV}}_{\text{alternate}} = \frac{1}{2}
\begin{bmatrix}
\cosh^2 r ~\mathbb{I}_2  & \sinh^2 r ~\mathbb{I}_2\\
\sinh^2 r ~\mathbb{I}_2  & \cosh^2 r ~\mathbb{I}_2 
\end{bmatrix},
\end{eqnarray}
respectively.

\bibliography{bib}

\begin{thebibliography}{83}
\expandafter\ifx\csname natexlab\endcsname\relax\def\natexlab#1{#1}\fi
\expandafter\ifx\csname bibnamefont\endcsname\relax
  \def\bibnamefont#1{#1}\fi
\expandafter\ifx\csname bibfnamefont\endcsname\relax
  \def\bibfnamefont#1{#1}\fi
\expandafter\ifx\csname citenamefont\endcsname\relax
  \def\citenamefont#1{#1}\fi
\expandafter\ifx\csname url\endcsname\relax
  \def\url#1{\texttt{#1}}\fi
\expandafter\ifx\csname urlprefix\endcsname\relax\def\urlprefix{URL }\fi
\providecommand{\bibinfo}[2]{#2}
\providecommand{\eprint}[2][]{\url{#2}}

\bibitem[{\citenamefont{Horodecki et~al.}(2009)\citenamefont{Horodecki,
  Horodecki, Horodecki, and Horodecki}}]{hhhh}
\bibinfo{author}{\bibfnamefont{R.}~\bibnamefont{Horodecki}},
  \bibinfo{author}{\bibfnamefont{P.}~\bibnamefont{Horodecki}},
  \bibinfo{author}{\bibfnamefont{M.}~\bibnamefont{Horodecki}},
  \bibnamefont{and}
  \bibinfo{author}{\bibfnamefont{K.}~\bibnamefont{Horodecki}},
  \bibinfo{journal}{Rev. Mod. Phys.} \textbf{\bibinfo{volume}{81}},
  \bibinfo{pages}{865} (\bibinfo{year}{2009}),
  \urlprefix\url{https://link.aps.org/doi/10.1103/RevModPhys.81.865}.

\bibitem[{\citenamefont{Schrödinger}(1935)}]{schrodinger}
\bibinfo{author}{\bibfnamefont{E.}~\bibnamefont{Schrödinger}},
  \bibinfo{journal}{Mathematical Proceedings of the Cambridge Philosophical
  Society} \textbf{\bibinfo{volume}{31}}, \bibinfo{pages}{555–563}
  (\bibinfo{year}{1935}).

\bibitem[{\citenamefont{Bennett et~al.}(1993)\citenamefont{Bennett, Brassard,
  Cr\'epeau, Jozsa, Peres, and Wootters}}]{tele1}
\bibinfo{author}{\bibfnamefont{C.~H.} \bibnamefont{Bennett}},
  \bibinfo{author}{\bibfnamefont{G.}~\bibnamefont{Brassard}},
  \bibinfo{author}{\bibfnamefont{C.}~\bibnamefont{Cr\'epeau}},
  \bibinfo{author}{\bibfnamefont{R.}~\bibnamefont{Jozsa}},
  \bibinfo{author}{\bibfnamefont{A.}~\bibnamefont{Peres}}, \bibnamefont{and}
  \bibinfo{author}{\bibfnamefont{W.~K.} \bibnamefont{Wootters}},
  \bibinfo{journal}{Phys. Rev. Lett.} \textbf{\bibinfo{volume}{70}},
  \bibinfo{pages}{1895} (\bibinfo{year}{1993}),
  \urlprefix\url{https://link.aps.org/doi/10.1103/PhysRevLett.70.1895}.

\bibitem[{\citenamefont{Bouwmeester et~al.}(1997)\citenamefont{Bouwmeester,
  Pan, Mattle, Eibl, Weinfurter, and Zeilinger}}]{tele2}
\bibinfo{author}{\bibfnamefont{D.}~\bibnamefont{Bouwmeester}},
  \bibinfo{author}{\bibfnamefont{J.-W.} \bibnamefont{Pan}},
  \bibinfo{author}{\bibfnamefont{K.}~\bibnamefont{Mattle}},
  \bibinfo{author}{\bibfnamefont{M.}~\bibnamefont{Eibl}},
  \bibinfo{author}{\bibfnamefont{H.}~\bibnamefont{Weinfurter}},
  \bibnamefont{and}
  \bibinfo{author}{\bibfnamefont{A.}~\bibnamefont{Zeilinger}},
  \bibinfo{journal}{Nature} \textbf{\bibinfo{volume}{390}},
  \bibinfo{pages}{575} (\bibinfo{year}{1997}),
  \urlprefix\url{https://doi.org/10.1038/37539}.

\bibitem[{\citenamefont{Pirandola et~al.}(2015)\citenamefont{Pirandola, Eisert,
  Weedbrook, Furusawa, and Braunstein}}]{tele3}
\bibinfo{author}{\bibfnamefont{S.}~\bibnamefont{Pirandola}},
  \bibinfo{author}{\bibfnamefont{J.}~\bibnamefont{Eisert}},
  \bibinfo{author}{\bibfnamefont{C.}~\bibnamefont{Weedbrook}},
  \bibinfo{author}{\bibfnamefont{A.}~\bibnamefont{Furusawa}}, \bibnamefont{and}
  \bibinfo{author}{\bibfnamefont{S.~L.} \bibnamefont{Braunstein}},
  \bibinfo{journal}{Nature Photonics} \textbf{\bibinfo{volume}{9}},
  \bibinfo{pages}{641} (\bibinfo{year}{2015}),
  \urlprefix\url{https://doi.org/10.1038/nphoton.2015.154}.

\bibitem[{\citenamefont{Bennett and Wiesner}(1992)}]{dc}
\bibinfo{author}{\bibfnamefont{C.~H.} \bibnamefont{Bennett}} \bibnamefont{and}
  \bibinfo{author}{\bibfnamefont{S.~J.} \bibnamefont{Wiesner}},
  \bibinfo{journal}{Phys. Rev. Lett.} \textbf{\bibinfo{volume}{69}},
  \bibinfo{pages}{2881} (\bibinfo{year}{1992}),
  \urlprefix\url{https://link.aps.org/doi/10.1103/PhysRevLett.69.2881}.

\bibitem[{\citenamefont{Guo et~al.}(2019)\citenamefont{Guo, Liu, Li, and
  Guo}}]{dcrev}
\bibinfo{author}{\bibfnamefont{Y.}~\bibnamefont{Guo}},
  \bibinfo{author}{\bibfnamefont{B.-H.} \bibnamefont{Liu}},
  \bibinfo{author}{\bibfnamefont{C.-F.} \bibnamefont{Li}}, \bibnamefont{and}
  \bibinfo{author}{\bibfnamefont{G.-C.} \bibnamefont{Guo}},
  \bibinfo{journal}{Advanced Quantum Technologies}
  \textbf{\bibinfo{volume}{2}}, \bibinfo{pages}{1900011}
  (\bibinfo{year}{2019}),
  \urlprefix\url{https://doi.org/10.1002/qute.201900011}.

\bibitem[{\citenamefont{Ekert}(1991)}]{crypto1}
\bibinfo{author}{\bibfnamefont{A.~K.} \bibnamefont{Ekert}},
  \bibinfo{journal}{Phys. Rev. Lett.} \textbf{\bibinfo{volume}{67}},
  \bibinfo{pages}{661} (\bibinfo{year}{1991}),
  \urlprefix\url{https://link.aps.org/doi/10.1103/PhysRevLett.67.661}.

\bibitem[{\citenamefont{Bennett et~al.}(1992)\citenamefont{Bennett, Brassard,
  and Mermin}}]{crypto2}
\bibinfo{author}{\bibfnamefont{C.~H.} \bibnamefont{Bennett}},
  \bibinfo{author}{\bibfnamefont{G.}~\bibnamefont{Brassard}}, \bibnamefont{and}
  \bibinfo{author}{\bibfnamefont{N.~D.} \bibnamefont{Mermin}},
  \bibinfo{journal}{Phys. Rev. Lett.} \textbf{\bibinfo{volume}{68}},
  \bibinfo{pages}{557} (\bibinfo{year}{1992}),
  \urlprefix\url{https://link.aps.org/doi/10.1103/PhysRevLett.68.557}.

\bibitem[{\citenamefont{Wu et~al.}(2004)\citenamefont{Wu, Sarandy, and
  Lidar}}]{qpt1}
\bibinfo{author}{\bibfnamefont{L.-A.} \bibnamefont{Wu}},
  \bibinfo{author}{\bibfnamefont{M.~S.} \bibnamefont{Sarandy}},
  \bibnamefont{and} \bibinfo{author}{\bibfnamefont{D.~A.} \bibnamefont{Lidar}},
  \bibinfo{journal}{Phys. Rev. Lett.} \textbf{\bibinfo{volume}{93}},
  \bibinfo{pages}{250404} (\bibinfo{year}{2004}),
  \urlprefix\url{https://link.aps.org/doi/10.1103/PhysRevLett.93.250404}.

\bibitem[{\citenamefont{Verstraete et~al.}(2004)\citenamefont{Verstraete,
  Mart\'{\i}n-Delgado, and Cirac}}]{qpt2}
\bibinfo{author}{\bibfnamefont{F.}~\bibnamefont{Verstraete}},
  \bibinfo{author}{\bibfnamefont{M.~A.} \bibnamefont{Mart\'{\i}n-Delgado}},
  \bibnamefont{and} \bibinfo{author}{\bibfnamefont{J.~I.} \bibnamefont{Cirac}},
  \bibinfo{journal}{Phys. Rev. Lett.} \textbf{\bibinfo{volume}{92}},
  \bibinfo{pages}{087201} (\bibinfo{year}{2004}),
  \urlprefix\url{https://link.aps.org/doi/10.1103/PhysRevLett.92.087201}.

\bibitem[{\citenamefont{Chakrabarti et~al.}(1996)\citenamefont{Chakrabarti,
  Dutta, and Sen}}]{qptbook1}
\bibinfo{author}{\bibfnamefont{B.~K.} \bibnamefont{Chakrabarti}},
  \bibinfo{author}{\bibfnamefont{A.}~\bibnamefont{Dutta}}, \bibnamefont{and}
  \bibinfo{author}{\bibfnamefont{P.}~\bibnamefont{Sen}},
  \emph{\bibinfo{title}{Quantum Ising phases and transitions in transverse
  Ising models}} (\bibinfo{publisher}{Springer Berlin Heidelberg},
  \bibinfo{year}{1996}), ISBN \bibinfo{isbn}{9783540498650},
  \urlprefix\url{https://doi.org/10.1007/978-3-540-49865-0}.

\bibitem[{\citenamefont{Sachdev}(2009)}]{qptbook2}
\bibinfo{author}{\bibfnamefont{S.}~\bibnamefont{Sachdev}},
  \emph{\bibinfo{title}{Quantum Phase Transitions}}
  (\bibinfo{publisher}{Cambridge University Press}, \bibinfo{year}{2009}),
  \urlprefix\url{https://doi.org/10.1017/cbo9780511973765}.

\bibitem[{\citenamefont{SHIMONY}(1995)}]{geo1}
\bibinfo{author}{\bibfnamefont{A.}~\bibnamefont{SHIMONY}},
  \bibinfo{journal}{Annals of the New York Academy of Sciences}
  \textbf{\bibinfo{volume}{755}}, \bibinfo{pages}{675} (\bibinfo{year}{1995}),
  \urlprefix\url{https://doi.org/10.1111/j.1749-6632.1995.tb39008.x}.

\bibitem[{\citenamefont{Brody and Hughston}(2001)}]{geo11}
\bibinfo{author}{\bibfnamefont{D.~C.} \bibnamefont{Brody}} \bibnamefont{and}
  \bibinfo{author}{\bibfnamefont{L.~P.} \bibnamefont{Hughston}},
  \bibinfo{journal}{Journal of Geometry and Physics}
  \textbf{\bibinfo{volume}{38}}, \bibinfo{pages}{19} (\bibinfo{year}{2001}),
  \urlprefix\url{https://doi.org/10.1016/s0393-0440(00)00052-8}.

\bibitem[{\citenamefont{Barnum and Linden}(2001{\natexlab{a}})}]{geo2}
\bibinfo{author}{\bibfnamefont{H.}~\bibnamefont{Barnum}} \bibnamefont{and}
  \bibinfo{author}{\bibfnamefont{N.}~\bibnamefont{Linden}},
  \bibinfo{journal}{Journal of Physics A: Mathematical and General}
  \textbf{\bibinfo{volume}{34}}, \bibinfo{pages}{6787}
  (\bibinfo{year}{2001}{\natexlab{a}}),
  \urlprefix\url{https://doi.org/10.1088/0305-4470/34/35/305}.

\bibitem[{\citenamefont{Meyer and Wallach}(2002)}]{geo3}
\bibinfo{author}{\bibfnamefont{D.~A.} \bibnamefont{Meyer}} \bibnamefont{and}
  \bibinfo{author}{\bibfnamefont{N.~R.} \bibnamefont{Wallach}},
  \bibinfo{journal}{Journal of Mathematical Physics}
  \textbf{\bibinfo{volume}{43}}, \bibinfo{pages}{4273} (\bibinfo{year}{2002}),
  \urlprefix\url{https://doi.org/10.1063/1.1497700}.

\bibitem[{\citenamefont{Wei and Goldbart}(2003{\natexlab{a}})}]{geo5}
\bibinfo{author}{\bibfnamefont{T.-C.} \bibnamefont{Wei}} \bibnamefont{and}
  \bibinfo{author}{\bibfnamefont{P.~M.} \bibnamefont{Goldbart}},
  \bibinfo{journal}{Phys. Rev. A} \textbf{\bibinfo{volume}{68}},
  \bibinfo{pages}{042307} (\bibinfo{year}{2003}{\natexlab{a}}),
  \urlprefix\url{https://link.aps.org/doi/10.1103/PhysRevA.68.042307}.

\bibitem[{\citenamefont{Osterloh and Siewert}(2005)}]{geo4}
\bibinfo{author}{\bibfnamefont{A.}~\bibnamefont{Osterloh}} \bibnamefont{and}
  \bibinfo{author}{\bibfnamefont{J.}~\bibnamefont{Siewert}},
  \bibinfo{journal}{Phys. Rev. A} \textbf{\bibinfo{volume}{72}},
  \bibinfo{pages}{012337} (\bibinfo{year}{2005}),
  \urlprefix\url{https://link.aps.org/doi/10.1103/PhysRevA.72.012337}.

\bibitem[{\citenamefont{Or\'us}(2008{\natexlab{a}})}]{geo6}
\bibinfo{author}{\bibfnamefont{R.}~\bibnamefont{Or\'us}},
  \bibinfo{journal}{Phys. Rev. Lett.} \textbf{\bibinfo{volume}{100}},
  \bibinfo{pages}{130502} (\bibinfo{year}{2008}{\natexlab{a}}),
  \urlprefix\url{https://link.aps.org/doi/10.1103/PhysRevLett.100.130502}.

\bibitem[{\citenamefont{Or\'us}(2008{\natexlab{b}})}]{geo7}
\bibinfo{author}{\bibfnamefont{R.}~\bibnamefont{Or\'us}},
  \bibinfo{journal}{Phys. Rev. A} \textbf{\bibinfo{volume}{78}},
  \bibinfo{pages}{062332} (\bibinfo{year}{2008}{\natexlab{b}}),
  \urlprefix\url{https://link.aps.org/doi/10.1103/PhysRevA.78.062332}.

\bibitem[{\citenamefont{Dokovi{\'{c}} and Osterloh}(2009)}]{geo8}
\bibinfo{author}{\bibfnamefont{D.~{\v{Z}}.} \bibnamefont{Dokovi{\'{c}}}}
  \bibnamefont{and} \bibinfo{author}{\bibfnamefont{A.}~\bibnamefont{Osterloh}},
  \bibinfo{journal}{Journal of Mathematical Physics}
  \textbf{\bibinfo{volume}{50}}, \bibinfo{pages}{033509}
  (\bibinfo{year}{2009}), \urlprefix\url{https://doi.org/10.1063/1.3075830}.

\bibitem[{\citenamefont{Shi et~al.}(2010)\citenamefont{Shi, Or{\'{u}}s,
  Fj{\ae}restad, and Zhou}}]{geo9}
\bibinfo{author}{\bibfnamefont{Q.-Q.} \bibnamefont{Shi}},
  \bibinfo{author}{\bibfnamefont{R.}~\bibnamefont{Or{\'{u}}s}},
  \bibinfo{author}{\bibfnamefont{J.~O.} \bibnamefont{Fj{\ae}restad}},
  \bibnamefont{and} \bibinfo{author}{\bibfnamefont{H.-Q.} \bibnamefont{Zhou}},
  \bibinfo{journal}{New Journal of Physics} \textbf{\bibinfo{volume}{12}},
  \bibinfo{pages}{025008} (\bibinfo{year}{2010}),
  \urlprefix\url{https://doi.org/10.1088/1367-2630/12/2/025008}.

\bibitem[{\citenamefont{Or\'us and Wei}(2010)}]{geo10}
\bibinfo{author}{\bibfnamefont{R.}~\bibnamefont{Or\'us}} \bibnamefont{and}
  \bibinfo{author}{\bibfnamefont{T.-C.} \bibnamefont{Wei}},
  \bibinfo{journal}{Phys. Rev. B} \textbf{\bibinfo{volume}{82}},
  \bibinfo{pages}{155120} (\bibinfo{year}{2010}),
  \urlprefix\url{https://link.aps.org/doi/10.1103/PhysRevB.82.155120}.

\bibitem[{\citenamefont{Coffman et~al.}(2000)\citenamefont{Coffman, Kundu, and
  Wootters}}]{gm1}
\bibinfo{author}{\bibfnamefont{V.}~\bibnamefont{Coffman}},
  \bibinfo{author}{\bibfnamefont{J.}~\bibnamefont{Kundu}}, \bibnamefont{and}
  \bibinfo{author}{\bibfnamefont{W.~K.} \bibnamefont{Wootters}},
  \bibinfo{journal}{Phys. Rev. A} \textbf{\bibinfo{volume}{61}},
  \bibinfo{pages}{052306} (\bibinfo{year}{2000}),
  \urlprefix\url{https://link.aps.org/doi/10.1103/PhysRevA.61.052306}.

\bibitem[{\citenamefont{Osborne and Verstraete}(2006)}]{gm2}
\bibinfo{author}{\bibfnamefont{T.~J.} \bibnamefont{Osborne}} \bibnamefont{and}
  \bibinfo{author}{\bibfnamefont{F.}~\bibnamefont{Verstraete}},
  \bibinfo{journal}{Phys. Rev. Lett.} \textbf{\bibinfo{volume}{96}},
  \bibinfo{pages}{220503} (\bibinfo{year}{2006}),
  \urlprefix\url{https://link.aps.org/doi/10.1103/PhysRevLett.96.220503}.

\bibitem[{\citenamefont{Adesso et~al.}(2006)\citenamefont{Adesso, Serafini, and
  Illuminati}}]{gm3}
\bibinfo{author}{\bibfnamefont{G.}~\bibnamefont{Adesso}},
  \bibinfo{author}{\bibfnamefont{A.}~\bibnamefont{Serafini}}, \bibnamefont{and}
  \bibinfo{author}{\bibfnamefont{F.}~\bibnamefont{Illuminati}},
  \bibinfo{journal}{Phys. Rev. A} \textbf{\bibinfo{volume}{73}},
  \bibinfo{pages}{032345} (\bibinfo{year}{2006}),
  \urlprefix\url{https://link.aps.org/doi/10.1103/PhysRevA.73.032345}.

\bibitem[{\citenamefont{Hiroshima et~al.}(2007)\citenamefont{Hiroshima, Adesso,
  and Illuminati}}]{gm4}
\bibinfo{author}{\bibfnamefont{T.}~\bibnamefont{Hiroshima}},
  \bibinfo{author}{\bibfnamefont{G.}~\bibnamefont{Adesso}}, \bibnamefont{and}
  \bibinfo{author}{\bibfnamefont{F.}~\bibnamefont{Illuminati}},
  \bibinfo{journal}{Phys. Rev. Lett.} \textbf{\bibinfo{volume}{98}},
  \bibinfo{pages}{050503} (\bibinfo{year}{2007}),
  \urlprefix\url{https://link.aps.org/doi/10.1103/PhysRevLett.98.050503}.

\bibitem[{\citenamefont{Ou and Fan}(2007)}]{gm5}
\bibinfo{author}{\bibfnamefont{Y.-C.} \bibnamefont{Ou}} \bibnamefont{and}
  \bibinfo{author}{\bibfnamefont{H.}~\bibnamefont{Fan}},
  \bibinfo{journal}{Phys. Rev. A} \textbf{\bibinfo{volume}{75}},
  \bibinfo{pages}{062308} (\bibinfo{year}{2007}),
  \urlprefix\url{https://link.aps.org/doi/10.1103/PhysRevA.75.062308}.

\bibitem[{\citenamefont{Kay et~al.}(2009)\citenamefont{Kay, Kaszlikowski, and
  Ramanathan}}]{gm6}
\bibinfo{author}{\bibfnamefont{A.}~\bibnamefont{Kay}},
  \bibinfo{author}{\bibfnamefont{D.}~\bibnamefont{Kaszlikowski}},
  \bibnamefont{and}
  \bibinfo{author}{\bibfnamefont{R.}~\bibnamefont{Ramanathan}},
  \bibinfo{journal}{Phys. Rev. Lett.} \textbf{\bibinfo{volume}{103}},
  \bibinfo{pages}{050501} (\bibinfo{year}{2009}),
  \urlprefix\url{https://link.aps.org/doi/10.1103/PhysRevLett.103.050501}.

\bibitem[{\citenamefont{Hayashi and Chen}(2011)}]{gm7}
\bibinfo{author}{\bibfnamefont{M.}~\bibnamefont{Hayashi}} \bibnamefont{and}
  \bibinfo{author}{\bibfnamefont{L.}~\bibnamefont{Chen}},
  \bibinfo{journal}{Phys. Rev. A} \textbf{\bibinfo{volume}{84}},
  \bibinfo{pages}{012325} (\bibinfo{year}{2011}),
  \urlprefix\url{https://link.aps.org/doi/10.1103/PhysRevA.84.012325}.

\bibitem[{\citenamefont{Fanchini et~al.}(2011)\citenamefont{Fanchini, Cornelio,
  de~Oliveira, and Caldeira}}]{gm8}
\bibinfo{author}{\bibfnamefont{F.~F.} \bibnamefont{Fanchini}},
  \bibinfo{author}{\bibfnamefont{M.~F.} \bibnamefont{Cornelio}},
  \bibinfo{author}{\bibfnamefont{M.~C.} \bibnamefont{de~Oliveira}},
  \bibnamefont{and} \bibinfo{author}{\bibfnamefont{A.~O.}
  \bibnamefont{Caldeira}}, \bibinfo{journal}{Phys. Rev. A}
  \textbf{\bibinfo{volume}{84}}, \bibinfo{pages}{012313}
  (\bibinfo{year}{2011}),
  \urlprefix\url{https://link.aps.org/doi/10.1103/PhysRevA.84.012313}.

\bibitem[{\citenamefont{Streltsov et~al.}(2012)\citenamefont{Streltsov, Adesso,
  Piani, and Bru\ss{}}}]{gm9}
\bibinfo{author}{\bibfnamefont{A.}~\bibnamefont{Streltsov}},
  \bibinfo{author}{\bibfnamefont{G.}~\bibnamefont{Adesso}},
  \bibinfo{author}{\bibfnamefont{M.}~\bibnamefont{Piani}}, \bibnamefont{and}
  \bibinfo{author}{\bibfnamefont{D.}~\bibnamefont{Bru\ss{}}},
  \bibinfo{journal}{Phys. Rev. Lett.} \textbf{\bibinfo{volume}{109}},
  \bibinfo{pages}{050503} (\bibinfo{year}{2012}),
  \urlprefix\url{https://link.aps.org/doi/10.1103/PhysRevLett.109.050503}.

\bibitem[{\citenamefont{Kumar et~al.}(2015)\citenamefont{Kumar, Prabhu,
  Sen(De), and Sen}}]{gm10}
\bibinfo{author}{\bibfnamefont{A.}~\bibnamefont{Kumar}},
  \bibinfo{author}{\bibfnamefont{R.}~\bibnamefont{Prabhu}},
  \bibinfo{author}{\bibfnamefont{A.}~\bibnamefont{Sen(De)}}, \bibnamefont{and}
  \bibinfo{author}{\bibfnamefont{U.}~\bibnamefont{Sen}},
  \bibinfo{journal}{Phys. Rev. A} \textbf{\bibinfo{volume}{91}},
  \bibinfo{pages}{012341} (\bibinfo{year}{2015}),
  \urlprefix\url{https://link.aps.org/doi/10.1103/PhysRevA.91.012341}.

\bibitem[{\citenamefont{Roy et~al.}(2018{\natexlab{a}})\citenamefont{Roy, Das,
  Kumar, Sen(De), and Sen}}]{gm11}
\bibinfo{author}{\bibfnamefont{S.}~\bibnamefont{Roy}},
  \bibinfo{author}{\bibfnamefont{T.}~\bibnamefont{Das}},
  \bibinfo{author}{\bibfnamefont{A.}~\bibnamefont{Kumar}},
  \bibinfo{author}{\bibfnamefont{A.}~\bibnamefont{Sen(De)}}, \bibnamefont{and}
  \bibinfo{author}{\bibfnamefont{U.}~\bibnamefont{Sen}},
  \bibinfo{journal}{Phys. Rev. A} \textbf{\bibinfo{volume}{98}},
  \bibinfo{pages}{012310} (\bibinfo{year}{2018}{\natexlab{a}}),
  \urlprefix\url{https://link.aps.org/doi/10.1103/PhysRevA.98.012310}.

\bibitem[{\citenamefont{Rethinasamy et~al.}(2019)\citenamefont{Rethinasamy,
  Roy, Chanda, Sen(De), and Sen}}]{gm12}
\bibinfo{author}{\bibfnamefont{S.}~\bibnamefont{Rethinasamy}},
  \bibinfo{author}{\bibfnamefont{S.}~\bibnamefont{Roy}},
  \bibinfo{author}{\bibfnamefont{T.}~\bibnamefont{Chanda}},
  \bibinfo{author}{\bibfnamefont{A.}~\bibnamefont{Sen(De)}}, \bibnamefont{and}
  \bibinfo{author}{\bibfnamefont{U.}~\bibnamefont{Sen}},
  \bibinfo{journal}{Phys. Rev. A} \textbf{\bibinfo{volume}{99}},
  \bibinfo{pages}{042302} (\bibinfo{year}{2019}),
  \urlprefix\url{https://link.aps.org/doi/10.1103/PhysRevA.99.042302}.

\bibitem[{\citenamefont{Sen(De) and Sen}(2010)}]{ggm1}
\bibinfo{author}{\bibfnamefont{A.}~\bibnamefont{Sen(De)}} \bibnamefont{and}
  \bibinfo{author}{\bibfnamefont{U.}~\bibnamefont{Sen}},
  \bibinfo{journal}{Phys. Rev. A} \textbf{\bibinfo{volume}{81}},
  \bibinfo{pages}{012308} (\bibinfo{year}{2010}),
  \urlprefix\url{https://link.aps.org/doi/10.1103/PhysRevA.81.012308}.

\bibitem[{\citenamefont{Shimony}(1995)}]{ggm22}
\bibinfo{author}{\bibfnamefont{A.}~\bibnamefont{Shimony}},
  \bibinfo{journal}{Annals of the New York Academy of Sciences}
  \textbf{\bibinfo{volume}{755}}, \bibinfo{pages}{675} (\bibinfo{year}{1995}),
  \urlprefix\url{https://doi.org/10.1111/j.1749-6632.1995.tb39008.x}.

\bibitem[{\citenamefont{Barnum and Linden}(2001{\natexlab{b}})}]{ggm3}
\bibinfo{author}{\bibfnamefont{H.}~\bibnamefont{Barnum}} \bibnamefont{and}
  \bibinfo{author}{\bibfnamefont{N.}~\bibnamefont{Linden}},
  \bibinfo{journal}{J. Phys. A: Mathematical and General}
  \textbf{\bibinfo{volume}{34}}, \bibinfo{pages}{6787}
  (\bibinfo{year}{2001}{\natexlab{b}}),
  \urlprefix\url{http://stacks.iop.org/0305-4470/34/i=35/a=305}.

\bibitem[{\citenamefont{Wei and Goldbart}(2003{\natexlab{b}})}]{ggm4}
\bibinfo{author}{\bibfnamefont{T.-C.} \bibnamefont{Wei}} \bibnamefont{and}
  \bibinfo{author}{\bibfnamefont{P.~M.} \bibnamefont{Goldbart}},
  \bibinfo{journal}{Phys. Rev. A} \textbf{\bibinfo{volume}{68}},
  \bibinfo{pages}{042307} (\bibinfo{year}{2003}{\natexlab{b}}),
  \urlprefix\url{https://link.aps.org/doi/10.1103/PhysRevA.68.042307}.

\bibitem[{\citenamefont{Blasone et~al.}(2008)\citenamefont{Blasone, Dell'Anno,
  De~Siena, and Illuminati}}]{ggm5}
\bibinfo{author}{\bibfnamefont{M.}~\bibnamefont{Blasone}},
  \bibinfo{author}{\bibfnamefont{F.}~\bibnamefont{Dell'Anno}},
  \bibinfo{author}{\bibfnamefont{S.}~\bibnamefont{De~Siena}}, \bibnamefont{and}
  \bibinfo{author}{\bibfnamefont{F.}~\bibnamefont{Illuminati}},
  \bibinfo{journal}{Phys. Rev. A} \textbf{\bibinfo{volume}{77}},
  \bibinfo{pages}{062304} (\bibinfo{year}{2008}),
  \urlprefix\url{https://link.aps.org/doi/10.1103/PhysRevA.77.062304}.

\bibitem[{\citenamefont{Prabhu et~al.}(2011)\citenamefont{Prabhu, Pradhan,
  Sen(De), and Sen}}]{gs1}
\bibinfo{author}{\bibfnamefont{R.}~\bibnamefont{Prabhu}},
  \bibinfo{author}{\bibfnamefont{S.}~\bibnamefont{Pradhan}},
  \bibinfo{author}{\bibfnamefont{A.}~\bibnamefont{Sen(De)}}, \bibnamefont{and}
  \bibinfo{author}{\bibfnamefont{U.}~\bibnamefont{Sen}},
  \bibinfo{journal}{Phys. Rev. A} \textbf{\bibinfo{volume}{84}},
  \bibinfo{pages}{042334} (\bibinfo{year}{2011}),
  \urlprefix\url{https://link.aps.org/doi/10.1103/PhysRevA.84.042334}.

\bibitem[{\citenamefont{Dhar et~al.}(2013)\citenamefont{Dhar, Sen(De), and
  Sen}}]{gs2}
\bibinfo{author}{\bibfnamefont{H.~S.} \bibnamefont{Dhar}},
  \bibinfo{author}{\bibfnamefont{A.}~\bibnamefont{Sen(De)}}, \bibnamefont{and}
  \bibinfo{author}{\bibfnamefont{U.}~\bibnamefont{Sen}},
  \bibinfo{journal}{Phys. Rev. Lett.} \textbf{\bibinfo{volume}{111}},
  \bibinfo{pages}{070501} (\bibinfo{year}{2013}),
  \urlprefix\url{https://link.aps.org/doi/10.1103/PhysRevLett.111.070501}.

\bibitem[{\citenamefont{Jindal et~al.}(2014)\citenamefont{Jindal, Rane, Dhar,
  Sen(De), and Sen}}]{gs4}
\bibinfo{author}{\bibfnamefont{L.}~\bibnamefont{Jindal}},
  \bibinfo{author}{\bibfnamefont{A.~D.} \bibnamefont{Rane}},
  \bibinfo{author}{\bibfnamefont{H.~S.} \bibnamefont{Dhar}},
  \bibinfo{author}{\bibfnamefont{A.}~\bibnamefont{Sen(De)}}, \bibnamefont{and}
  \bibinfo{author}{\bibfnamefont{U.}~\bibnamefont{Sen}},
  \bibinfo{journal}{Phys. Rev. A} \textbf{\bibinfo{volume}{89}},
  \bibinfo{pages}{012316} (\bibinfo{year}{2014}),
  \urlprefix\url{https://link.aps.org/doi/10.1103/PhysRevA.89.012316}.

\bibitem[{\citenamefont{Biswas et~al.}(2014)\citenamefont{Biswas, Prabhu,
  Sen(De), and Sen}}]{gs3}
\bibinfo{author}{\bibfnamefont{A.}~\bibnamefont{Biswas}},
  \bibinfo{author}{\bibfnamefont{R.}~\bibnamefont{Prabhu}},
  \bibinfo{author}{\bibfnamefont{A.}~\bibnamefont{Sen(De)}}, \bibnamefont{and}
  \bibinfo{author}{\bibfnamefont{U.}~\bibnamefont{Sen}},
  \bibinfo{journal}{Phys. Rev. A} \textbf{\bibinfo{volume}{90}},
  \bibinfo{pages}{032301} (\bibinfo{year}{2014}),
  \urlprefix\url{https://link.aps.org/doi/10.1103/PhysRevA.90.032301}.

\bibitem[{\citenamefont{Das et~al.}(2014)\citenamefont{Das, Prabhu, Sen(De),
  and Sen}}]{gdc1}
\bibinfo{author}{\bibfnamefont{T.}~\bibnamefont{Das}},
  \bibinfo{author}{\bibfnamefont{R.}~\bibnamefont{Prabhu}},
  \bibinfo{author}{\bibfnamefont{A.}~\bibnamefont{Sen(De)}}, \bibnamefont{and}
  \bibinfo{author}{\bibfnamefont{U.}~\bibnamefont{Sen}},
  \bibinfo{journal}{Phys. Rev. A} \textbf{\bibinfo{volume}{90}},
  \bibinfo{pages}{022319} (\bibinfo{year}{2014}),
  \urlprefix\url{https://link.aps.org/doi/10.1103/PhysRevA.90.022319}.

\bibitem[{\citenamefont{Das et~al.}(2016{\natexlab{a}})\citenamefont{Das, Roy,
  Bagchi, Misra, Sen(De), and Sen}}]{ggm2}
\bibinfo{author}{\bibfnamefont{T.}~\bibnamefont{Das}},
  \bibinfo{author}{\bibfnamefont{S.~S.} \bibnamefont{Roy}},
  \bibinfo{author}{\bibfnamefont{S.}~\bibnamefont{Bagchi}},
  \bibinfo{author}{\bibfnamefont{A.}~\bibnamefont{Misra}},
  \bibinfo{author}{\bibfnamefont{A.}~\bibnamefont{Sen(De)}}, \bibnamefont{and}
  \bibinfo{author}{\bibfnamefont{U.}~\bibnamefont{Sen}},
  \bibinfo{journal}{Phys. Rev. A} \textbf{\bibinfo{volume}{94}},
  \bibinfo{pages}{022336} (\bibinfo{year}{2016}{\natexlab{a}}),
  \urlprefix\url{https://link.aps.org/doi/10.1103/PhysRevA.94.022336}.

\bibitem[{\citenamefont{Braunstein and van Loock}(2005)}]{contvar-rev}
\bibinfo{author}{\bibfnamefont{S.~L.} \bibnamefont{Braunstein}}
  \bibnamefont{and} \bibinfo{author}{\bibfnamefont{P.}~\bibnamefont{van
  Loock}}, \bibinfo{journal}{Rev. Mod. Phys.} \textbf{\bibinfo{volume}{77}},
  \bibinfo{pages}{513} (\bibinfo{year}{2005}),
  \urlprefix\url{https://link.aps.org/doi/10.1103/RevModPhys.77.513}.

\bibitem[{\citenamefont{Braunstein and Pati}(2003)}]{contvar-book}
\bibinfo{editor}{\bibfnamefont{S.~L.} \bibnamefont{Braunstein}}
  \bibnamefont{and} \bibinfo{editor}{\bibfnamefont{A.~K.} \bibnamefont{Pati}},
  eds., \emph{\bibinfo{title}{Quantum Information with Continuous Variables}}
  (\bibinfo{publisher}{Springer Netherlands}, \bibinfo{year}{2003}),
  \urlprefix\url{https://doi.org/10.1007/978-94-015-1258-9}.

\bibitem[{\citenamefont{Adesso et~al.}(2014)\citenamefont{Adesso, Ragy, and
  Lee}}]{gaussian3}
\bibinfo{author}{\bibfnamefont{G.}~\bibnamefont{Adesso}},
  \bibinfo{author}{\bibfnamefont{S.}~\bibnamefont{Ragy}}, \bibnamefont{and}
  \bibinfo{author}{\bibfnamefont{A.~R.} \bibnamefont{Lee}},
  \bibinfo{journal}{Open Systems {\&} Information Dynamics}
  \textbf{\bibinfo{volume}{21}}, \bibinfo{pages}{1440001}
  (\bibinfo{year}{2014}),
  \urlprefix\url{https://doi.org/10.1142/s1230161214400010}.

\bibitem[{\citenamefont{Weedbrook et~al.}(2012)\citenamefont{Weedbrook,
  Pirandola, Garc\'{\i}a-Patr\'on, Cerf, Ralph, Shapiro, and
  Lloyd}}]{gaussian1}
\bibinfo{author}{\bibfnamefont{C.}~\bibnamefont{Weedbrook}},
  \bibinfo{author}{\bibfnamefont{S.}~\bibnamefont{Pirandola}},
  \bibinfo{author}{\bibfnamefont{R.}~\bibnamefont{Garc\'{\i}a-Patr\'on}},
  \bibinfo{author}{\bibfnamefont{N.~J.} \bibnamefont{Cerf}},
  \bibinfo{author}{\bibfnamefont{T.~C.} \bibnamefont{Ralph}},
  \bibinfo{author}{\bibfnamefont{J.~H.} \bibnamefont{Shapiro}},
  \bibnamefont{and} \bibinfo{author}{\bibfnamefont{S.}~\bibnamefont{Lloyd}},
  \bibinfo{journal}{Rev. Mod. Phys.} \textbf{\bibinfo{volume}{84}},
  \bibinfo{pages}{621} (\bibinfo{year}{2012}),
  \urlprefix\url{https://link.aps.org/doi/10.1103/RevModPhys.84.621}.

\bibitem[{\citenamefont{WANG et~al.}(2007)\citenamefont{WANG, HIROSHIMA,
  TOMITA, and HAYASHI}}]{gaussian2}
\bibinfo{author}{\bibfnamefont{X.}~\bibnamefont{WANG}},
  \bibinfo{author}{\bibfnamefont{T.}~\bibnamefont{HIROSHIMA}},
  \bibinfo{author}{\bibfnamefont{A.}~\bibnamefont{TOMITA}}, \bibnamefont{and}
  \bibinfo{author}{\bibfnamefont{M.}~\bibnamefont{HAYASHI}},
  \bibinfo{journal}{Physics Reports} \textbf{\bibinfo{volume}{448}},
  \bibinfo{pages}{1} (\bibinfo{year}{2007}),
  \urlprefix\url{https://doi.org/10.1016/j.physrep.2007.04.005}.

\bibitem[{\citenamefont{Chuang et~al.}(1997)\citenamefont{Chuang, Leung, and
  Yamamoto}}]{bc1}
\bibinfo{author}{\bibfnamefont{I.~L.} \bibnamefont{Chuang}},
  \bibinfo{author}{\bibfnamefont{D.~W.} \bibnamefont{Leung}}, \bibnamefont{and}
  \bibinfo{author}{\bibfnamefont{Y.}~\bibnamefont{Yamamoto}},
  \bibinfo{journal}{Phys. Rev. A} \textbf{\bibinfo{volume}{56}},
  \bibinfo{pages}{1114} (\bibinfo{year}{1997}),
  \urlprefix\url{https://link.aps.org/doi/10.1103/PhysRevA.56.1114}.

\bibitem[{\citenamefont{Lloyd and Braunstein}(1999)}]{qcomp}
\bibinfo{author}{\bibfnamefont{S.}~\bibnamefont{Lloyd}} \bibnamefont{and}
  \bibinfo{author}{\bibfnamefont{S.~L.} \bibnamefont{Braunstein}},
  \bibinfo{journal}{Phys. Rev. Lett.} \textbf{\bibinfo{volume}{82}},
  \bibinfo{pages}{1784} (\bibinfo{year}{1999}),
  \urlprefix\url{https://link.aps.org/doi/10.1103/PhysRevLett.82.1784}.

\bibitem[{\citenamefont{Huver et~al.}(2008)\citenamefont{Huver, Wildfeuer, and
  Dowling}}]{qmetro}
\bibinfo{author}{\bibfnamefont{S.~D.} \bibnamefont{Huver}},
  \bibinfo{author}{\bibfnamefont{C.~F.} \bibnamefont{Wildfeuer}},
  \bibnamefont{and} \bibinfo{author}{\bibfnamefont{J.~P.}
  \bibnamefont{Dowling}}, \bibinfo{journal}{Phys. Rev. A}
  \textbf{\bibinfo{volume}{78}}, \bibinfo{pages}{063828}
  (\bibinfo{year}{2008}),
  \urlprefix\url{https://link.aps.org/doi/10.1103/PhysRevA.78.063828}.

\bibitem[{\citenamefont{Eisert et~al.}(2002)\citenamefont{Eisert, Scheel, and
  Plenio}}]{Ent-distillation}
\bibinfo{author}{\bibfnamefont{J.}~\bibnamefont{Eisert}},
  \bibinfo{author}{\bibfnamefont{S.}~\bibnamefont{Scheel}}, \bibnamefont{and}
  \bibinfo{author}{\bibfnamefont{M.~B.} \bibnamefont{Plenio}},
  \bibinfo{journal}{Phys. Rev. Lett.} \textbf{\bibinfo{volume}{89}},
  \bibinfo{pages}{137903} (\bibinfo{year}{2002}),
  \urlprefix\url{https://link.aps.org/doi/10.1103/PhysRevLett.89.137903}.

\bibitem[{\citenamefont{Giedke and
  Ignacio~Cirac}(2002)}]{Giedke-ent-distillation}
\bibinfo{author}{\bibfnamefont{G.}~\bibnamefont{Giedke}} \bibnamefont{and}
  \bibinfo{author}{\bibfnamefont{J.}~\bibnamefont{Ignacio~Cirac}},
  \bibinfo{journal}{Phys. Rev. A} \textbf{\bibinfo{volume}{66}},
  \bibinfo{pages}{032316} (\bibinfo{year}{2002}),
  \urlprefix\url{https://link.aps.org/doi/10.1103/PhysRevA.66.032316}.

\bibitem[{\citenamefont{Sabapathy et~al.}(2011)\citenamefont{Sabapathy, Ivan,
  and Simon}}]{entdistribution}
\bibinfo{author}{\bibfnamefont{K.~K.} \bibnamefont{Sabapathy}},
  \bibinfo{author}{\bibfnamefont{J.~S.} \bibnamefont{Ivan}}, \bibnamefont{and}
  \bibinfo{author}{\bibfnamefont{R.}~\bibnamefont{Simon}},
  \bibinfo{journal}{Phys. Rev. Lett.} \textbf{\bibinfo{volume}{107}},
  \bibinfo{pages}{130501} (\bibinfo{year}{2011}),
  \urlprefix\url{https://link.aps.org/doi/10.1103/PhysRevLett.107.130501}.

\bibitem[{\citenamefont{Niset et~al.}(2009)\citenamefont{Niset,
  Fiur\'a\ifmmode~\check{s}\else \v{s}\fi{}ek, and Cerf}}]{gaussian-error}
\bibinfo{author}{\bibfnamefont{J.}~\bibnamefont{Niset}},
  \bibinfo{author}{\bibfnamefont{J.}~\bibnamefont{Fiur\'a\ifmmode~\check{s}\else
  \v{s}\fi{}ek}}, \bibnamefont{and} \bibinfo{author}{\bibfnamefont{N.~J.}
  \bibnamefont{Cerf}}, \bibinfo{journal}{Phys. Rev. Lett.}
  \textbf{\bibinfo{volume}{102}}, \bibinfo{pages}{120501}
  (\bibinfo{year}{2009}),
  \urlprefix\url{https://link.aps.org/doi/10.1103/PhysRevLett.102.120501}.

\bibitem[{\citenamefont{Adesso et~al.}(2009)\citenamefont{Adesso, Dell'Anno,
  De~Siena, Illuminati, and Souza}}]{phase-estimation}
\bibinfo{author}{\bibfnamefont{G.}~\bibnamefont{Adesso}},
  \bibinfo{author}{\bibfnamefont{F.}~\bibnamefont{Dell'Anno}},
  \bibinfo{author}{\bibfnamefont{S.}~\bibnamefont{De~Siena}},
  \bibinfo{author}{\bibfnamefont{F.}~\bibnamefont{Illuminati}},
  \bibnamefont{and} \bibinfo{author}{\bibfnamefont{L.~A.~M.}
  \bibnamefont{Souza}}, \bibinfo{journal}{Phys. Rev. A}
  \textbf{\bibinfo{volume}{79}}, \bibinfo{pages}{040305}
  (\bibinfo{year}{2009}),
  \urlprefix\url{https://link.aps.org/doi/10.1103/PhysRevA.79.040305}.

\bibitem[{\citenamefont{Zhang and van Loock}(2011)}]{qcommunication}
\bibinfo{author}{\bibfnamefont{S.}~\bibnamefont{Zhang}} \bibnamefont{and}
  \bibinfo{author}{\bibfnamefont{P.}~\bibnamefont{van Loock}},
  \bibinfo{journal}{Phys. Rev. A} \textbf{\bibinfo{volume}{84}},
  \bibinfo{pages}{062309} (\bibinfo{year}{2011}),
  \urlprefix\url{https://link.aps.org/doi/10.1103/PhysRevA.84.062309}.

\bibitem[{\citenamefont{Cerf et~al.}(2005)\citenamefont{Cerf, Kr\"uger, Navez,
  Werner, and Wolf}}]{qcloning}
\bibinfo{author}{\bibfnamefont{N.~J.} \bibnamefont{Cerf}},
  \bibinfo{author}{\bibfnamefont{O.}~\bibnamefont{Kr\"uger}},
  \bibinfo{author}{\bibfnamefont{P.}~\bibnamefont{Navez}},
  \bibinfo{author}{\bibfnamefont{R.~F.} \bibnamefont{Werner}},
  \bibnamefont{and} \bibinfo{author}{\bibfnamefont{M.~M.} \bibnamefont{Wolf}},
  \bibinfo{journal}{Phys. Rev. Lett.} \textbf{\bibinfo{volume}{95}},
  \bibinfo{pages}{070501} (\bibinfo{year}{2005}),
  \urlprefix\url{https://link.aps.org/doi/10.1103/PhysRevLett.95.070501}.

\bibitem[{\citenamefont{Wigner}(1932)}]{Wigner}
\bibinfo{author}{\bibfnamefont{E.}~\bibnamefont{Wigner}},
  \bibinfo{journal}{Phys. Rev.} \textbf{\bibinfo{volume}{40}},
  \bibinfo{pages}{749} (\bibinfo{year}{1932}),
  \urlprefix\url{https://link.aps.org/doi/10.1103/PhysRev.40.749}.

\bibitem[{\citenamefont{{Usha Devi, A. R.} et~al.}(2006)\citenamefont{{Usha
  Devi, A. R.}, {Prabhu, R.}, and {Uma, M. S.}}}]{prabhu-usha}
\bibinfo{author}{\bibnamefont{{Usha Devi, A. R.}}},
  \bibinfo{author}{\bibnamefont{{Prabhu, R.}}}, \bibnamefont{and}
  \bibinfo{author}{\bibnamefont{{Uma, M. S.}}}, \bibinfo{journal}{Eur. Phys. J.
  D} \textbf{\bibinfo{volume}{40}}, \bibinfo{pages}{133}
  (\bibinfo{year}{2006}),
  \urlprefix\url{https://doi.org/10.1140/epjd/e2006-00135-x}.

\bibitem[{\citenamefont{Navarrete-Benlloch
  et~al.}(2012)\citenamefont{Navarrete-Benlloch, Garc\'{\i}a-Patr\'on, Shapiro,
  and Cerf}}]{TMSV}
\bibinfo{author}{\bibfnamefont{C.}~\bibnamefont{Navarrete-Benlloch}},
  \bibinfo{author}{\bibfnamefont{R.}~\bibnamefont{Garc\'{\i}a-Patr\'on}},
  \bibinfo{author}{\bibfnamefont{J.~H.} \bibnamefont{Shapiro}},
  \bibnamefont{and} \bibinfo{author}{\bibfnamefont{N.~J.} \bibnamefont{Cerf}},
  \bibinfo{journal}{Phys. Rev. A} \textbf{\bibinfo{volume}{86}},
  \bibinfo{pages}{012328} (\bibinfo{year}{2012}),
  \urlprefix\url{https://link.aps.org/doi/10.1103/PhysRevA.86.012328}.

\bibitem[{\citenamefont{Das et~al.}(2016{\natexlab{b}})\citenamefont{Das,
  Prabhu, Sen(De), and Sen}}]{FMSV}
\bibinfo{author}{\bibfnamefont{T.}~\bibnamefont{Das}},
  \bibinfo{author}{\bibfnamefont{R.}~\bibnamefont{Prabhu}},
  \bibinfo{author}{\bibfnamefont{A.}~\bibnamefont{Sen(De)}}, \bibnamefont{and}
  \bibinfo{author}{\bibfnamefont{U.}~\bibnamefont{Sen}},
  \bibinfo{journal}{Phys. Rev. A} \textbf{\bibinfo{volume}{93}},
  \bibinfo{pages}{052313} (\bibinfo{year}{2016}{\natexlab{b}}),
  \urlprefix\url{https://link.aps.org/doi/10.1103/PhysRevA.93.052313}.

\bibitem[{\citenamefont{Walschaers et~al.}(2019)\citenamefont{Walschaers, Ra,
  and Treps}}]{exp1}
\bibinfo{author}{\bibfnamefont{M.}~\bibnamefont{Walschaers}},
  \bibinfo{author}{\bibfnamefont{Y.-S.} \bibnamefont{Ra}}, \bibnamefont{and}
  \bibinfo{author}{\bibfnamefont{N.}~\bibnamefont{Treps}},
  \bibinfo{journal}{Phys. Rev. A} \textbf{\bibinfo{volume}{100}},
  \bibinfo{pages}{023828} (\bibinfo{year}{2019}),
  \urlprefix\url{https://link.aps.org/doi/10.1103/PhysRevA.100.023828}.

\bibitem[{\citenamefont{Ra et~al.}(2019)\citenamefont{Ra, Dufour, Walschaers,
  Jacquard, Michel, Fabre, and Treps}}]{exp2}
\bibinfo{author}{\bibfnamefont{Y.-S.} \bibnamefont{Ra}},
  \bibinfo{author}{\bibfnamefont{A.}~\bibnamefont{Dufour}},
  \bibinfo{author}{\bibfnamefont{M.}~\bibnamefont{Walschaers}},
  \bibinfo{author}{\bibfnamefont{C.}~\bibnamefont{Jacquard}},
  \bibinfo{author}{\bibfnamefont{T.}~\bibnamefont{Michel}},
  \bibinfo{author}{\bibfnamefont{C.}~\bibnamefont{Fabre}}, \bibnamefont{and}
  \bibinfo{author}{\bibfnamefont{N.}~\bibnamefont{Treps}},
  \bibinfo{journal}{Nature Physics} \textbf{\bibinfo{volume}{16}},
  \bibinfo{pages}{144} (\bibinfo{year}{2019}),
  \urlprefix\url{https://doi.org/10.1038/s41567-019-0726-y}.

\bibitem[{\citenamefont{Arnold}(1989)}]{fubini1}
\bibinfo{author}{\bibfnamefont{V.~I.} \bibnamefont{Arnold}},
  \emph{\bibinfo{title}{Mathematical Methods of Classical Mechanics}}
  (\bibinfo{publisher}{Springer New York}, \bibinfo{year}{1989}),
  \urlprefix\url{https://doi.org/10.1007/978-1-4757-2063-1}.

\bibitem[{\citenamefont{Kobayashi}(1996)}]{fubini2}
\bibinfo{author}{\bibfnamefont{S.}~\bibnamefont{Kobayashi}},
  \emph{\bibinfo{title}{Foundations of differential geometry}}
  (\bibinfo{publisher}{Wiley}, \bibinfo{address}{New York},
  \bibinfo{year}{1996}), ISBN \bibinfo{isbn}{978-0-471-15732-8}.

\bibitem[{\citenamefont{Botero and Reznik}(2003)}]{schmidt}
\bibinfo{author}{\bibfnamefont{A.}~\bibnamefont{Botero}} \bibnamefont{and}
  \bibinfo{author}{\bibfnamefont{B.}~\bibnamefont{Reznik}},
  \bibinfo{journal}{Phys. Rev. A} \textbf{\bibinfo{volume}{67}},
  \bibinfo{pages}{052311} (\bibinfo{year}{2003}),
  \urlprefix\url{https://link.aps.org/doi/10.1103/PhysRevA.67.052311}.

\bibitem[{\citenamefont{Herbut}(2018)}]{schmidt2}
\bibinfo{author}{\bibfnamefont{F.}~\bibnamefont{Herbut}},
  \bibinfo{journal}{Quanta} \textbf{\bibinfo{volume}{7}}, \bibinfo{pages}{19}
  (\bibinfo{year}{2018}),
  \urlprefix\url{https://doi.org/10.12743/quanta.v7i1.69}.

\bibitem[{fer(2005)}]{ferraro}
\emph{\bibinfo{title}{Gaussian States in Quantum Information}}
  (\bibinfo{publisher}{Bibliopolis}, \bibinfo{year}{2005}), ISBN
  \bibinfo{isbn}{887088483X},
  \urlprefix\url{https://www.xarg.org/ref/a/887088483X/}.

\bibitem[{\citenamefont{Buchholz et~al.}(2015)\citenamefont{Buchholz, Moroder,
  and G\"{u}hne}}]{gmixed}
\bibinfo{author}{\bibfnamefont{L.~E.} \bibnamefont{Buchholz}},
  \bibinfo{author}{\bibfnamefont{T.}~\bibnamefont{Moroder}}, \bibnamefont{and}
  \bibinfo{author}{\bibfnamefont{O.}~\bibnamefont{G\"{u}hne}},
  \bibinfo{journal}{Annalen der Physik} \textbf{\bibinfo{volume}{528}},
  \bibinfo{pages}{278} (\bibinfo{year}{2015}),
  \urlprefix\url{https://doi.org/10.1002/andp.201500293}.

\bibitem[{\citenamefont{Bondani et~al.}(2004)\citenamefont{Bondani, Allevi,
  Puddu, Andreoni, Ferraro, and Paris}}]{ferraro2}
\bibinfo{author}{\bibfnamefont{M.}~\bibnamefont{Bondani}},
  \bibinfo{author}{\bibfnamefont{A.}~\bibnamefont{Allevi}},
  \bibinfo{author}{\bibfnamefont{E.}~\bibnamefont{Puddu}},
  \bibinfo{author}{\bibfnamefont{A.}~\bibnamefont{Andreoni}},
  \bibinfo{author}{\bibfnamefont{A.}~\bibnamefont{Ferraro}}, \bibnamefont{and}
  \bibinfo{author}{\bibfnamefont{M.~G.~A.} \bibnamefont{Paris}},
  \bibinfo{journal}{Optics Letters} \textbf{\bibinfo{volume}{29}},
  \bibinfo{pages}{180} (\bibinfo{year}{2004}),
  \urlprefix\url{https://doi.org/10.1364/ol.29.000180}.

\bibitem[{\citenamefont{Yang}(2016)}]{enhance1}
\bibinfo{author}{\bibfnamefont{Y.}~\bibnamefont{Yang}},
  \bibinfo{journal}{Journal of the Optical Society of America B}
  \textbf{\bibinfo{volume}{33}}, \bibinfo{pages}{2545} (\bibinfo{year}{2016}),
  \urlprefix\url{https://doi.org/10.1364/josab.33.002545}.

\bibitem[{\citenamefont{Roy et~al.}(2018{\natexlab{b}})\citenamefont{Roy,
  Chanda, Das, Sen(De), and Sen}}]{bell}
\bibinfo{author}{\bibfnamefont{S.}~\bibnamefont{Roy}},
  \bibinfo{author}{\bibfnamefont{T.}~\bibnamefont{Chanda}},
  \bibinfo{author}{\bibfnamefont{T.}~\bibnamefont{Das}},
  \bibinfo{author}{\bibfnamefont{A.}~\bibnamefont{Sen(De)}}, \bibnamefont{and}
  \bibinfo{author}{\bibfnamefont{U.}~\bibnamefont{Sen}},
  \bibinfo{journal}{Phys. Rev. A} \textbf{\bibinfo{volume}{98}},
  \bibinfo{pages}{052131} (\bibinfo{year}{2018}{\natexlab{b}}),
  \urlprefix\url{https://link.aps.org/doi/10.1103/PhysRevA.98.052131}.

\bibitem[{\citenamefont{Mojaveri et~al.}(2019)\citenamefont{Mojaveri, Dehghani,
  and Bahrbeig}}]{enhance3}
\bibinfo{author}{\bibfnamefont{B.}~\bibnamefont{Mojaveri}},
  \bibinfo{author}{\bibfnamefont{A.}~\bibnamefont{Dehghani}}, \bibnamefont{and}
  \bibinfo{author}{\bibfnamefont{R.~J.} \bibnamefont{Bahrbeig}},
  \bibinfo{journal}{The European Physical Journal Plus}
  \textbf{\bibinfo{volume}{134}} (\bibinfo{year}{2019}),
  \urlprefix\url{https://doi.org/10.1140/epjp/i2019-12823-7}.

\bibitem[{\citenamefont{Duc and Dat}(2020)}]{enhance2}
\bibinfo{author}{\bibfnamefont{T.~M.} \bibnamefont{Duc}} \bibnamefont{and}
  \bibinfo{author}{\bibfnamefont{T.~Q.} \bibnamefont{Dat}},
  \bibinfo{journal}{Optik} p. \bibinfo{pages}{164479} (\bibinfo{year}{2020}),
  \urlprefix\url{https://doi.org/10.1016/j.ijleo.2020.164479}.

\bibitem[{\citenamefont{Genoni et~al.}(2008)\citenamefont{Genoni, Paris, and
  Banaszek}}]{ng1}
\bibinfo{author}{\bibfnamefont{M.~G.} \bibnamefont{Genoni}},
  \bibinfo{author}{\bibfnamefont{M.~G.~A.} \bibnamefont{Paris}},
  \bibnamefont{and} \bibinfo{author}{\bibfnamefont{K.}~\bibnamefont{Banaszek}},
  \bibinfo{journal}{Phys. Rev. A} \textbf{\bibinfo{volume}{78}},
  \bibinfo{pages}{060303} (\bibinfo{year}{2008}),
  \urlprefix\url{https://link.aps.org/doi/10.1103/PhysRevA.78.060303}.

\bibitem[{\citenamefont{Genoni and Paris}(2010)}]{ng2}
\bibinfo{author}{\bibfnamefont{M.~G.} \bibnamefont{Genoni}} \bibnamefont{and}
  \bibinfo{author}{\bibfnamefont{M.~G.~A.} \bibnamefont{Paris}},
  \bibinfo{journal}{Phys. Rev. A} \textbf{\bibinfo{volume}{82}},
  \bibinfo{pages}{052341} (\bibinfo{year}{2010}),
  \urlprefix\url{https://link.aps.org/doi/10.1103/PhysRevA.82.052341}.

\bibitem[{\citenamefont{Marian and Marian}(2013)}]{ng3}
\bibinfo{author}{\bibfnamefont{P.}~\bibnamefont{Marian}} \bibnamefont{and}
  \bibinfo{author}{\bibfnamefont{T.~A.} \bibnamefont{Marian}},
  \bibinfo{journal}{Phys. Rev. A} \textbf{\bibinfo{volume}{88}},
  \bibinfo{pages}{012322} (\bibinfo{year}{2013}),
  \urlprefix\url{https://link.aps.org/doi/10.1103/PhysRevA.88.012322}.

\bibitem[{\citenamefont{Bartley and Walmsley}(2015)}]{Bartley-2015}
\bibinfo{author}{\bibfnamefont{T.~J.} \bibnamefont{Bartley}} \bibnamefont{and}
  \bibinfo{author}{\bibfnamefont{I.~A.} \bibnamefont{Walmsley}},
  \bibinfo{journal}{New Journal of Physics} \textbf{\bibinfo{volume}{17}},
  \bibinfo{pages}{023038} (\bibinfo{year}{2015}),
  \urlprefix\url{https://doi.org/10.1088%2F1367-2630%2F17%2F2%2F023038}.

\end{thebibliography}

%\begin{thebibliography}{99}
%
%
%\bibitem{ger-rev} G. Adesso, S. Ragy, and A. R. Lee, Open Syst. Inf. Dyn. {\bf 21}, 1440001 (2014);  	arXiv:1401.4679 [quant-ph].
%
%\end{thebibliography}
\end{document}